\documentclass[12pt]{article}
\usepackage[margin=1in]{geometry}

\usepackage{graphicx}%
\usepackage{multirow}%
\usepackage{amsmath,amssymb,amsfonts}%
\usepackage{amsthm}%
\usepackage{mathrsfs}%
\usepackage[title]{appendix}%
\usepackage{xcolor}%
\usepackage{textcomp}%
\usepackage{manyfoot}%
\usepackage{booktabs}%
\usepackage{algorithm}%
\usepackage{algorithmicx}%
\usepackage{algpseudocode}%
\usepackage{listings}%

\usepackage{geometry}
\usepackage{enumerate}
\usepackage{subcaption}
\usepackage{tikz}
\usepackage{hyperref}
\usepackage[version=3]{mhchem}
\usepackage{array}

\theoremstyle{thmstyleone}%
\newtheorem{theorem}{Theorem}[section]
\newtheorem{corollary}[theorem]{Corollary}
\newtheorem{proposition}[theorem]{Proposition}
\newtheorem{lemma}[theorem]{Lemma}

\theoremstyle{thmstyletwo}%
\newtheorem{remark}[theorem]{Remark}

\theoremstyle{thmstylethree}%
\newtheorem{definition}{Definition}%

\newcommand{\R}{\mathbb{R}}
\newcommand{\Rnn}{\mathbb{R}_{\geq 0}}

\newcommand{\cone}{\ker(S)\cap\Rnn^r}
\newcommand{\LAM}{\Lambda\left(E\right)}
\DeclareMathOperator{\diag}{diag}
\DeclareMathOperator{\rank}{rank}
\DeclareMathOperator{\sign}{sign}

\begin{document}

\title{In distributive phosphorylation catalytic constants enable
  non-trivial dynamics}

\author{Carsten Conradi\footnote{Hochschule für Technik und
    Wirtschaft, Berlin, Germany. carsten.conradi@htw-berlin.de}, Maya
  Mincheva\footnote{Department of Mathematical Sciences, Northern
    Illinois University, DeKalb, IL, USA. mmincheva@niu.edu}}

\maketitle

\begin{abstract}
  Ordered distributive double phosphorylation is a recurrent motif
  in intracellular signaling and control. It is either sequential
  (where the site phosphorylated last is dephosphorylated first) or
  cyclic (where the site phosphorylated first is dephosphorylated
  first). Sequential distributive double phosphorylation has been
  extensively studied and an inequality involving only the catalytic
  constants of kinase and phosphatase is known to be sufficient for
  multistationarity. As multistationarity is necessary for
  bistability it has been argued that these constants enable 
  bistability.

  Here we show for cyclic distributive double
  phosphorylation that if its catalytic constants satisfy the {\em
    very same} inequality, then Hopf bifurcations and hence
  sustained oscillations can occur. Hence we argue that in
  distributive double phosphorylation (sequential or distributive)
  the catalytic constants enable non-trivial dynamics.

  In fact, if the rate constant values in a network of cyclic
  distributive double phosphorylation are such that Hopf
  bifurcations and sustained oscillations can occur, then a network
  of sequential distributive double phosphorylation with the {\em
    same rate constant values} will show multistationarity -- albeit
  for different values of the total concentrations. For
  cyclic distributive double phosphorylation we further describe a 
  procedure to generate rate constant values where Hopf
  bifurcations and hence sustained oscillations can occur. This may,
  for example, allow for an efficient sampling of oscillatory
  regions in parameter space.

  Our analysis is greatly simplified by the fact that it is possible to
  reduce the network of cyclic distributive double phosphorylation to
  what we call a network with a single extreme ray. We summarize key
  properties of these networks.

\end{abstract}

\section{Introduction}

Phosphorylation is a process where proteins are altered by adding
and removing phosphate groups at designated binding sites. It is a  
recurrent motif in many large reaction networks involved in
intracellular signaling and control \cite{Suwanmajo2018}. Often
spatial effects are neglected and the time dynamics of the
participating chemical species concentrations is described by ordinary
differential equations. There exists a plethora of small
ODE models that include phosphorylation at one or two binding sites
together with the interaction of various regulating chemical species,
see, for example,~\cite{Ramesh2023}. These models exhibit a wide range of
dynamical properties ranging from multistationarity and bistability to
sustained oscillations~\cite{Ramesh2023}.

Phosphorylation and dephosphorylation are catalyzed by two enzymes, a
kinase and a phosphatase. 
As described in~\cite{sig-034,sig-041},
this process can either be processive or distributive:
if it is processive, then all available binding sites are 
phosphorylated or de-phosphorylated upon binding of protein and
kinase/phosphatase. If the process is distributive, then at most one
binding site is modified upon each binding of protein and kinase
or phosphatase.
Furthermore, multisite phosphorylation and dephosphorylation can occur
at 
a random sequence of
binding sites or at an ordered sequence of binding
sites. An ordered mechanism is either sequential or cyclic
\cite{sig-034,sig-041}. In a sequential mechanism the last 
site to be phosphorylated is dephosphorylated first, while in a cyclic
mechanism the first site to be phosphorylated is also dephosphorylated
first
-- as depicted in the reaction schemes of
Fig.~\ref{fig:cyc_d}.

\begin{figure}[!h]
  \centering

  \begin{subfigure}{0.45\textwidth}
    \centering
    \includegraphics[width=0.95\textwidth]{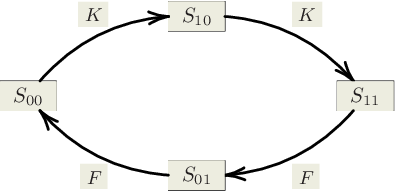}
    \subcaption{
      \label{fig:scheme_cyc}
      Cyclic \& distributive 
    }
  \end{subfigure}
  \begin{subfigure}{0.45\textwidth}
    \centering
    \includegraphics[width=0.95\textwidth]{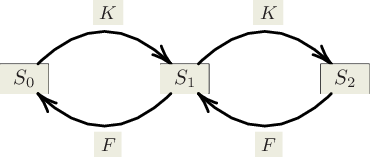}
    \caption{
      \label{fig:scheme_dis}
      Sequential \& distributive 
    }
  \end{subfigure}
  
  \caption{
    \label{fig:cyc_d}
    Reaction schemes describing distributive double phosphorylation 
    of a protein $S$ at two binding sites by a kinase $K$ and a
    phosphatase $F$. Panel~(\subref{fig:scheme_cyc}) cyclic double
    phosphorylation, panel~(\subref{fig:scheme_dis}) distributive. In
    panel ~(\subref{fig:scheme_cyc}) the subscript $_{ij}$ denotes the
    state of the phosphorylation sites:~$0$ unphosphorylated, $1$
    phosphorylated  (e.g.~$S_{10}$ denotes those molecules of $S$,
    where the first site is phosphorylated and the second site is
    unphosphorylated). In  panel~(\subref{fig:scheme_dis}) the
    subscript $_i$ denotes the number of attached phosphate groups
    (e.g.~$S_0$ denotes unphosphorylated protein).
  }
\end{figure}

Networks of sequential phosphorylation have been studied extensively and
it has been shown that already networks without any form of regulation
can exhibit non-trivial dynamics. It is, for example, known that
models of sequential and  processive phosphorylation and
dephosphorylation have a unique, globally attracting steady
state~\cite{processive-Nsite}. For sequential
distributive phosphorylation and dephosphorylation
multistationarity and bistability have been established
\cite{Conradi2008,Hell2015,Holstein2013}. And for a mixed 
mechanism of sequential distributive phosphorylation and processive
dephosphorylation at two binding sites Hopf bifurcations and sustained 
oscillations have been described
in~\cite{Suwanmajo2015} and~\cite{osc-008}.
For more 
information about the dynamics of multisite phosphorylation systems
see the review~\cite{ptm-028}  and the references therein.

Networks of cyclic phosphorylation have been studied
in~\cite{ptm-034}. The authors study in particular the following
mass action network~(\ref{eq:MA_cyc}) derived from  the reaction
scheme of Fig.~\ref{fig:scheme_cyc} (where the notation is the same
as in Fig.~\ref{fig:scheme_cyc} and $K S_{00}$, $K S_{10}$, $F
S_{11}$ and $F S_{01}$ denote enzyme-substrate complexes):
\begin{equation} 
  \begin{split}
    \label{eq:MA_cyc}
    S_{00} + K \ce{<=>[\kappa_1][\kappa_2]} KS_{00}
    \ce{->[\kappa_3]} S_{10}+K 
    \ce{<=>[\kappa_4][\kappa_5]} KS_{10} \ce{->[\kappa_6]} S_{11}+K \\
    S_{11} + F  \ce{<=>[\kappa_7][\kappa_8]} FS_{11}
    \ce{->[\kappa_9]} S_{01}+F \ce{<=>[\kappa_{10}][\kappa_{11}]}
    FS_{01} \ce{->[\kappa_{12}]} S_{00}+F\ .
  \end{split}
\end{equation}
In~\cite{ptm-034} it is shown that for every admissible positive 
value of the total concentrations and all positive values of the
reaction rate constants there exists a unique positive steady
state. The authors furthermore provide parameter values where a Hopf
bifurcation occurs (with the total amount of kinase as a bifurcation
parameter) and sustained oscillations emerge
\cite[Fig.~4]{ptm-034}. The authors however neither describe how the
respective parameter values were found nor do they explain how 
parameter values leading to oscillations can be found.

Here we study the same network~(\ref{eq:MA_cyc}) and provide
an answer to the latter question. In particular, we show that if the
rate constants $\kappa_3$, $\kappa_6$, $\kappa_9$ and $\kappa_{12}$
satisfy the inequality 
\begin{equation}
  \label{eq:ineq_multi_hopf}
  \kappa_3\, \kappa_9 - \kappa_6\, \kappa_{12} <0,
\end{equation}
then Hopf-bifurcations and sustained oscillations can occur. In
enzyme catalysis these constants are known as \emph{catalytic
  constants} (cf.\ the discussion in Section~\ref{sec:discussion} or
\cite{maya-bistab}). 

To be
more precise, we analyze the following irreversible subnetwork of
(\ref{eq:MA_cyc}) 
\begin{equation}
  \begin{split}
    \label{net:cyc_simple}    
    S_{00} + K \ce{->[\kappa_1]} K S_{00}  \ce{->[\kappa_3]} S_{10}+K
    \ce{->[\kappa_4]} K S_{10}  \ce{->[\kappa_6]} S_{11}+K \\
    S_{11} + F  \ce{->[\kappa_7]} F S_{11}  \ce{->[\kappa_9]} S_{01}+F
    \ce{->[\kappa_{10}]} F S_{01}  \ce{->[\kappa_{12}]} S_{00}+F
  \end{split}
\end{equation}
and show that if the values of the catalytic
constants $\kappa_3$, $\kappa_6$, $\kappa_9$ and $\kappa_{12}$ 
satisfy~(\ref{eq:ineq_multi_hopf}), then there exist values for
$\kappa_1$, $\kappa_4$, $\kappa_7$, $\kappa_{10}$ and the total
concentrations such that the steady state Jacobian of
network~(\ref{net:cyc_simple}) has a complex-conjugate pair of
eigenvalues on the imaginary axis. Furthermore, if this eigenvalue
pair crosses the imaginary axis as one of the parameters is varied,
then a Hopf bifurcation occurs at this particular steady state
(Theorem~\ref{thm:stab}). We describe a derivative condition for
this crossing and a procedure to find such parameter values. Based
on a result by \cite{Banaji2018}, we then argue that if a
supercritical Hopf bifurcation occurs in
network~(\ref{net:cyc_simple}) and a stable limit cycle emerges,
then the full network~(\ref{eq:MA_cyc}) will have a stable limit
cycle as well (for appropriately chosen values of $\kappa_2$,
$\kappa_5$, $\kappa_8$ and $\kappa_{11}$). This is
Theorem~\ref{theo:main-cyc-full}.

Finally, we compare our results for {\em cyclic} distributive double
phosphorylation with those obtained in \cite{maya-bistab} for
{\em sequential} distributive double phosphorylation. In
\cite{maya-bistab} an inequality that is sufficient for
multistationarity has been described. Remarkably this inequality
involves the catalytic constants in the same way as our
inequality~(\ref{eq:ineq_multi_hopf}). As multistationarity is
necessary for bistability it is argued in~\cite{maya-bistab} that
these rate constants enable bistability in sequential and
distributive double phosphorylation. Hence we conclude that in
distributive phosphorylation (sequential or cyclic) these constants
enable non-trivial dynamics.

The paper is organized as follows: 
to arrive at our results the ODEs derived from
network~(\ref{eq:MA_cyc}) and (\ref{net:cyc_simple}) are analyzed. For
this purpose we recall in Section~\ref{sec:biochem_rn_general} some
well known facts about ODEs defined by reaction networks and introduce
a special class of reaction networks that we call networks with a
single extreme ray. In this section we also derive conditions for
a simple Hopf bifurcation in such networks.
In Section~\ref{sec:analysis-network-cyc-simple} we first analyze
network~(\ref{net:cyc_simple}) and verify these bifurcation
conditions in Theorem~\ref{thm:stab}. Then we turn to the full
network~(\ref{eq:MA_cyc}) and Theorem~\ref{theo:main-cyc-full}. We
also present the procedure to determine rate constant and total
concentration values. In Section~\ref{sec:discussion} we discuss
inequality~(\ref{eq:ineq_multi_hopf}) in the light of the 
the results presented 
in~\cite{maya-bistab}. Appendix~\ref{App:Rem_A_lamB}
and~\ref{sec:hur-hom} contain some of the longer proofs of the 
results presented in
Section~\ref{sec:biochem_rn_general}. Appendix~\ref{app:jac-h-lambda},
\ref{sec:initial_data} and \ref{sec:Initial_data_Intro_Exa_Seq}
contain information to reproduce the numerical results displayed 
in the figures throughout the paper.

\section{Biochemical reaction networks with mass action kinetics}
\label{sec:biochem_rn_general}

To establish our results  we exploit the special
structure of the Jacobian of
a certain class of reaction networks that we call {\em networks with a
  single extreme ray}. We introduce this class here in full
generality. To this end we first consider a general reaction network
with $n$ species and $r$ reactions in Section~\ref{sec:general_RN} and recall 
the structure of the ODEs defined by such a general network. In
Section~\ref{sec:steady-states} we discuss steady states and
formally define networks with a single extreme ray. In
Section~\ref{sec:general_J} we present a formula for the 
Jacobian of a general reaction network at steady state. In 
Section~\ref{sec:single_extreme_vector} we discuss the steady state
Jacobian of networks with a single extreme ray. Finally, in
Section~\ref{sec:Hopf_single_E} we present conditions for simple Hopf
bifurcations in networks of this kind.

\subsection{Reaction networks with $n$ species and $r$
  reactions} 
\label{sec:general_RN}

We briefly introduce the relevant notation, for a more detailed
discussion we refer to the large body of literature on mass action
networks, for example, \cite{fein-02} or \cite{Conradi2019}.

To every chemical species we associate a variable $x_i$ denoting its 
concentration. For network~(\ref{eq:MA_cyc}) and
(\ref{net:cyc_simple}) we use the association as
given in Table~\ref{tab:species}.

\begin{table}
  \centering
  \begin{tabular}{|c|c|c|c|c|c|c|c|c|c|}\hline
    $x_1$ & $x_2$ & $x_3$ & $x_4$ & $x_5$ & $x_6$ & $x_7$ & $x_8$ & $x_9$ & $x_{10}$ \\ \hline
    $K$ & $F$ & $S_{00}$ & $S_{10}$ & $S_{01}$ & $S_{11}$ & $K S_{00} $ & $KS_{10}$ & $FS_{01} $ & $FS_{11} $ \\ \hline
   \end{tabular}
  \caption{Species and concentration variables for
    network~(\ref{eq:MA_cyc}) and (\ref{net:cyc_simple}).}
  \label{tab:species}
\end{table}

Consider network~(\ref{eq:MA_cyc}), nodes like
$S_{00}+K$ are called complexes. To every complex we associate a
vector $y\in\R^n$ representing the stoichiometry of the associated
chemical species. The complex $S_{00}+K$ consists of one unit of
$S_{00}$ and one unit of $K$, hence, in the ordering of
Table~\ref{tab:species} the vector $y^T=(1,0,1,0,0,0,0,0,0,0)$
represents its stoichiometry. In a similar way one arrives at the ten
complex vectors for networks~(\ref{eq:MA_cyc}) and
(\ref{net:cyc_simple}) given in 
Table~\ref{tab:complexes-cyc-and-cyc-simple} of
Appendix~\ref{app:jac-h-lambda}.

To every reaction $r^{(l)}$ we associate the difference of the complexes at
the tip and the tail of the reaction arrow. For example, to the 
reaction $S_{00}+K \to S_{00}K$ we associate the vector
$r^{(l)}=y^{(2)}-y^{(1)} = (-1,-1,0,1,0,0,0,0,0,0)^T$  (using the labeling
of complexes given in Table~\ref{tab:complexes-cyc-and-cyc-simple} of
Appendix~\ref{app:jac-h-lambda}). All reaction vectors $r^{(i)}$ are
collected as columns of the stoichiometric matrix $S\in\R^{n\times
  r}$: 
\begin{equation}
  \label{eq:def_S}
  S=\left[
    \begin{array}{ccc}
      r^{(1)} & \dots & r^{(r)}
    \end{array}
  \right]\ .
\end{equation}

To every reaction $r^{(l)}$ we further associate a reaction rate
function $v_l(k,x)$ describing the \lq speed\rq{} of the reaction. We
consider only mass action kinetics, hence the reaction rate function
of reaction $r^{(l)}: y^{(i)}\to y^{(j)}$ is given by the monomial
function 
\begin{displaymath}
  v_l(k,x) = k_l\, \prod_{l=1}^n x_l^{y_l^{(i)}} = k_l\, x^{y^{(i)}},
\end{displaymath} 
where $k_l$ is a parameter called the rate constant and $x^y$ is the
customary shorthand notation for the product $\prod_{l=1}^n x_l^{y_l}$
of two $n$-vectors $x$ and $y$.

We collect the vectors at the tail of every reaction as columns of a
matrix $Y$ and note that this matrix may contain several copies of the
same vector:
\begin{equation}
  \label{eq:def_Y}
  Y=\left[
    \begin{array}{ccc}
      y^{(1)} & \dots & y^{(r)}
    \end{array}
  \right]
\end{equation}
We collect all rate constants in a vector $k^T=(k_1$, \ldots, $k_r)$
and the monomials $x^{y^{(i)}}$ in a vector 
\begin{equation}
  \label{eq:def_phi}
  \phi(x) = (x^{y^{(1)}}, \ldots, x^{y^{(r)}})^T\ ,
\end{equation}
where the vectors $y^{(i)}$ reference the columns of the
matrix~$Y$ defined in ~(\ref{eq:def_Y}). The reaction rate function is
then defined using $k$ and $\phi(x)$ as
\begin{equation}
  \label{eq:def_v}
  v(k,x) = \diag(k)\, \phi(x)\ .
\end{equation}
After an ordering of species and reactions is fixed, every reaction
network with mass action kinetics defines a matrix $S$ and a reaction
rate function $v(k,x)$ in a unique way. These objects in turn define
the following system of ODEs:
\begin{equation}
  \label{eq:def_S_v}
  \dot x = S v(k,x)\ ,
\end{equation}
where $\dot x$ denotes the vector of
derivatives with respect to time.
Example are abundant in the literature, see, for example,
\cite[Section~2]{fein-02} or \cite{Conradi2019}.

Often the matrix~$S$ does not have full row rank. In this case let $W$
be a matrix whose columns span $\ker(S^T)$, that is a full rank matrix
$W$ with
$W^T\, S \equiv 0$. 
Then
$W^T\, \dot x \equiv 0$
and for every solution $x(t)$ with initial value $x(0)=x_0$ one
obtains 
\begin{displaymath}
  W^T\, x(t) = W^T\, x_0 = \text{const. }
\end{displaymath}
Hence, if $\rank(S)=s<n$, then one obtains $n-s$ conservation
relations
\begin{equation}
  \label{eq:def_con_rel}
  W^T\, x = c\ .
\end{equation}

\subsection{Steady states}
\label{sec:steady-states}

We are interested in points ($k$,$x$) that are solutions of 
\begin{equation}
  \label{eq:def_ss}
  S v(k,x) = 0\ .
\end{equation}
If ($k$,$x$) is a solution of~(\ref{eq:def_ss}), then $x$ is a steady
state of (\ref{eq:def_S_v}) for the rate constants $k$.

We proceed as in~\cite{Conradi2020} and express the reaction 
rates $v(k,x)$ at a steady state as a nonnegative combination of the
extreme vectors of the pointed polyhedral cone $\ker(S)\cap\Rnn^r$.
This idea goes back to Clarke and coworkers (cf.\ for example,
\cite{Clarke1988,ptm-028}) and it is as follows: ($k$,$x$) satisfy
(\ref{eq:def_ss}), if and only if the corresponding $v(k,x)$ is such
that $v(k,x)\in\cone$ (cf.\ eg.~\cite{Conradi2020}).

Convex polyhedral cones have a finite number of extreme vectors (up to
a  scalar positive multiplication \cite{Rockafellar1970}). Therefore,
any element $v$ of such a cone can be represented as a nonnegative 
linear combination of its extreme vectors $\{E_1, \ldots  ,E_l  \}$ 
\begin{equation}\label{eq:convex-cone}
  v = \sum_{i=1}^l \lambda_i E_i = E \lambda,\; \lambda \in\Rnn^l,
\end{equation}
where $E$ is the matrix with columns $E_1,\dots,E_l$ and
$\lambda^T=(\lambda_1,\dots,\lambda_l)$.

\begin{remark}[The relative interior of $\cone$]
  \label{rem:rel_int_cone}
  \begin{enumerate}[{(}A{)}] 
  \item As explained in, for example, \cite{alg-048}, a
    system~(\ref{eq:def_S_v}) has {\em positive} solutions ($k$,$x$),
    if and only if the matrix $E$ does not contain a zero row. Hence we
    will only consider networks where the matrix~$E$ are of this kind.
  \item We are only interested in positive values of $k$ and $x$. Thus
    we are only interested in $v(k,x)$ that are strictly positive and
    hence belong to the {\em relative interior} of the cone
    $\cone$. Consequently we are only interested in those
    $\lambda\in\Rnn^l$ that yield a positive $E\lambda$. In the
    remainder of the paper we therefore restrict $\lambda\in\Rnn^l$ to
    the set:
    \begin{equation}
      \label{eq:def_LAM}
      \Lambda(E) := \left\{ \lambda\in\Rnn^l | E\lambda >0 \right\}\ .
    \end{equation}
    This set has been introduced in \cite{fein-02}, cf.\
    \cite[Remark~4]{fein-02}.
  \item A more detailed discussion of the relation between positive
    solutions of (\ref{eq:def_ss}), the cone $\cone$ and the generator
    matrix~$E$ can be found in \cite{alg-048}, cf.\ in particular
    \cite[Proposition~6]{alg-048}. 
  \end{enumerate}
\end{remark}

In the analysis of network~(\ref{eq:MA_cyc}) we will later analyze
subnetworks that fit the following definition:

\begin{definition}[Networks with a single extreme ray]
  Consider a reaction network with stoichiometric matrix~$S$. If the
  cone $\cone$ is generated by a single positive vector then we call
  it a {\em network with a single extreme ray}. 
\end{definition}

As discussed in Remark~\ref{rem:rel_int_cone}, if $x>0$ and $k>0$,
then $v(k,x)=\diag(k)\phi(x)>0$. Consequently, 
if ($k$,$x$) satisfy the steady state equation~(\ref{eq:def_ss}), then
$v(k,x)\in\cone$ and there exists a vector $\lambda\in\LAM$, such
that~(\ref{eq:convex-cone}) is satisfied. Thus one may parameterize
all reaction rates at steady states via the equation
\begin{equation}
  \label{eq:diag_k_phi_E_lam}
  \diag(k)\, \phi(x) = E\, \lambda,\; \lambda\in\LAM\ .
\end{equation}
Likewise, given some~$\lambda\in\LAM$ and a
positive~$x$ one obtains a positive vector~$k$ such that~($k$,$x$)
satisfy the steady state equation~(\ref{eq:def_ss}) by
the following formula:
\begin{equation}
  \label{eq:def_k}
  k = \diag(\phi(h))E\lambda,
\end{equation}
where
\begin{equation}
  \label{eq:def_h}
  h=\frac{1}{x} \; .
\end{equation}
We observe that (i) the formula~(\ref{eq:def_h}) is well defined as by
assumption~$x>0$, (ii) that the vector~$k$ from~(\ref{eq:def_k}) is
positive since all entries of $E\lambda$ are positive (as, by
assumption, the matrix~$E$ does not have any zero rows, cf.\
Remark~\ref{rem:rel_int_cone}) and (iii) that 
the formula~(\ref{eq:def_k}) is obtained by solving the $k$-linear
equation~(\ref{eq:diag_k_phi_E_lam}) for~$k$.

\subsection{The Jacobian at steady  state}
\label{sec:general_J}

The equations~(\ref{eq:convex-cone})~\&~(\ref{eq:diag_k_phi_E_lam})
introduce a parametrization of the reaction rates~$v(k,x)$ at steady
states in terms of
the convex parameters~$h$ and $\lambda$. This parameterization can be
used to parameterize the Jacobian at steady state: as explained in
detail in~\cite{Conradi2020}  or~\cite{Clarke1988}, if ($k$,$x$)
satisfy the steady state equation~(\ref{eq:def_ss}) and hence $v(k,x)$
is such that~(\ref{eq:diag_k_phi_E_lam}) holds, then the Jacobian
evaluated at that ($k$, $x$) is given by the following formula
(cf.\cite[Proposition~2]{Conradi2020}):
\begin{equation}
  \label{eq:new-jac}
  J(k,x) = J(h,\lambda) = N \diag(E\lambda) Y^T
  \diag(h),
\end{equation}
where~$Y$ denotes the matrix introduced
in~(\ref{eq:def_Y}).

\subsection{The Jacobian of networks with a single extreme
  vector} 
\label{sec:single_extreme_vector}
 
For any reaction network where the cone $\cone$ is spanned by a single
positive vector~$E$ the parameter $\lambda$ is a positive scalar and
formula~(\ref{eq:new-jac}) is equivalent to
\begin{equation}
  \label{eq:Jac_E_one-D}
  J_{\lambda}(h) = \lambda S\diag(E)Y^T\diag(h)\ .
\end{equation} 
We use $J_{\lambda}(h)$ to denote the Jacobian in this special case.  
 
Here we comment on the characteristic polynomial
of general matrices $J_{\lambda}(h)$, that is on the polynomial
$\det(\mu I - J_{\lambda}(h))$ of an $n\times n$ Jacobian of the
form~(\ref{eq:Jac_E_one-D}). We will assume that $\rank(J_{\lambda}) =s
\leq n$ and we will use the symbol $a_i(h,\lambda)$ to denote the
coefficients of its characteristic polynomial. We will also discuss
the case $\lambda=1$, that is, the characteristic polynomial $\det(\mu
I - J_1(h))$ of $J_1(h)$ with coefficients $b_i(h)$.
For this purpose, the characteristic polynomial of the Jacobian
$J_{\lambda} (h)$ is denoted by
\begin{equation}\label{cp:ai}
\det (\mu I -J_{\lambda} (h)) = \mu^{n-s}(\mu^s+ a_1 (\lambda, h)
\mu^{s-1} + \ldots + a_s(\lambda,h)),
\end{equation} 
while the characteristic polynomial of $J_{1} (h)$ is denoted by
\begin{equation}\label{cp:bi}
\det (\mu I -J_{1} (h)) = \mu^{n-s}(\mu^s+ b_1 (h) \mu^{s-1}
+ \ldots + b_s (h)) \; . 
\end{equation}
Concerning these polynomials we have the following
corollary of Lemma~\ref{lem:4} in Appendix~\ref{App:Rem_A_lamB} where
we show the relationship between the coefficients $a_i(h,\lambda)$ and
$b_i (h)$.
\begin{corollary}
  \label{coro:cp_J_lamH}
  Suppose the matrix $E$ consists of a single positive column vector. 
  Let $J_{\lambda}(h)$ be as in (\ref{eq:Jac_E_one-D}) with
  $\rank(J_{\lambda}(h))=\rank(J_1(h)) =s <n$ and let
  $a_i(\lambda,h)$, $b_i(h)$ be the coefficients of the
  characteristic polynomials 
  in~(\ref{cp:ai}) and 
  in~(\ref{cp:bi}). Then the coefficients $a_i(\lambda,h)$ and $b_i(h)$
  satisfy: 
  \begin{equation}
    \label{eq:ai_bi}
    a_i(\lambda,h) = \lambda^i b_i(h),\, i=1,\, \ldots,\, s\ . 
  \end{equation}
  Moreover, the polynomial $\det(\mu I - J_{\lambda}(h))$ is 
  given by the following formula:
  \begin{equation}
    \begin{split}
      \label{eq:char_poly_J_H}
      \det(\mu I_n - J_{\lambda} (h)) &= \det(\mu I_n-\lambda J_1 (h)) \\
      &= \mu^{n-s} \lambda^s
      \left(
        \left(\frac{\mu}{\lambda}\right)^s
        + \sum_{i=1}^s  b_i(h)
        \left(\frac{\mu}{\lambda}\right)^{s-i} 
      \right)
    \end{split}
  \end{equation}
\end{corollary}
\noindent
We observe the following consequences of
Corollary~\ref{coro:cp_J_lamH}:

\begin{remark}
  \label{rem:consequences_E_vector}
  \begin{enumerate}[{(}I{)}]
  \item By (\ref{eq:char_poly_J_H}) it follows that if $\omega(h)$ is an
    eigenvalue of $J_1(h)$, then $\mu(\lambda,\omega) = \lambda
    \omega(h)$ is an eigenvalue of $J_{\lambda}(h)$. 
  \item In our setting $\lambda>0$, hence 
    the sign of the real part of the eigenvalues $Re(\mu(\lambda,h))$
    of $J_{\lambda}(h)$ is independent of $\lambda$:
    $\sign(Re(\mu(\lambda,h))) = \sign(Re(\omega(h)))$.
  \item In particular,  the matrix $J_{\lambda}(h)$ has a purely
    imaginary pair of eigenvalues $\pm i \lambda \omega(h)$ if and
    only if $J_1 (h)$ has a purely imaginary pair $\pm i \omega(h)$. 
  \end{enumerate}
\end{remark}

\subsection{Simple Hopf bifurcations in networks with a single extreme ray}
\label{sec:Hopf_single_E}

We recall the definition of a {\em simple Hopf bifurcation} for a
parameter dependent system of ODEs of the form  $\dot{x} = g_p(x)$,  
where $x \in \R^s$, and $g_p(x)$ varies smoothly in~$p$ and~$x$.
Let $x^* \in \R^s$ be  a steady state of the ODE system for some fixed
value $p_0$, that is, $g_{p_0}(x^*)=0$. Furthermore,  we assume that
we have a smooth curve of steady states around~$p_0$:
\begin{displaymath}
  p ~\mapsto ~ x(p)~ \text{with $x(p_0)=x^*$.}
\end{displaymath}
That is, $g_{p}\left( x(p) \right)= 0$ for all $p$ close enough
to $p_0$.

Further let $J(x(p),p)$ be the Jacobian of $g_p(x)$ evaluated
at~$x(p)$. If, as $p$ varies, a complex-conjugate pair of eigenvalues
of $J(x(p),p)$ crosses the imaginary axis, then there exists a Hopf
bifurcation at $(x(p_0,p_0))$.
A {\em simple Hopf bifurcation} occurs at $(x(p_0),p_0)$, if no 
other eigenvalue crosses the imaginary axis at the same value~$p_0$. 
In this case, a limit cycle arises. If the 
Hopf bifurcation is supercritical, then stable periodic solutions are
generated for nearby parameter values~\cite{BifTheo-009}.

Similar to~\cite{osc-008} and \cite{Conradi2020} we want to build on
results described in \cite{BifTheo-009} and \cite{osc-011} to
establish Hopf bifurcations. 
As in previous work (cf.\
e.g.~\cite{Conradi2020,BifTheo-009,osc-011}), we will use a
criterion based on the following Hurwitz determinants:
\begin{definition} \label{def:hurwitz}
  The {\em $i$-th Hurwitz matrix} of a univariate polynomial 
  $p(z)= a_0 z^s + a_{1} z^{s-1} + \cdots + a_s$  
  is the following $i \times i$ matrix:
  \[
    H_i ~=~ 
    \begin{pmatrix}
      a_1 & a_0 & 0 & 0 & 0 & \cdots & 0 \\
      a_3 & a_2 & a_1 & a_0 & 0 & \cdots & 0 \\
      \vdots & \vdots & \vdots &\vdots & \vdots &  & \vdots \\
      a_{2i-1} & a_{2i-2} & a_{2i-3} & a_{2i-4} &a_{2i-5} &\cdots & a_i
    \end{pmatrix}~,
  \]
  in which the $(k,l)$-th entry is $a_{2k-l}$ 
  as long as $0\leq 2 k - l \leq s$, and
  $0$ otherwise.  The determinants $\det(H_i)$ are called 
  \emph{Hurwitz determinants}.
\end{definition}

To every square matrix one can associate Hurwitz matrices via its  
characteristic polynomial in an analogous way.
In the following we consider two families of Hurwitz matrices
constructed according to Definition~\ref{def:hurwitz}: matrices
$H_l(\lambda,h)$ obtained from coefficients $a_i(\lambda,h)$
of~(\ref{cp:ai}) and matrices $G_l(h)$ of coefficients $b_i(h)$
of~(\ref{cp:bi}). Our analysis is greatly simplified by the
relationship between the two families established in 
Proposition~\ref{pro:hom} below (a proof is given in
Appendix~\ref{sec:hur-hom}): 
\begin{proposition}\label{pro:hom}
  The Hurwitz determinants for the characteristic polynomials
  (\ref{cp:ai}) and (\ref{cp:bi}) satisfy the following equation
  \begin{displaymath}
    \det \left(H_{l} (\lambda, h)\right) = \lambda^{l(l+1)/2}  \det
    \left(G_{l} (h)\right), \; \; l=1,2, \ldots ,s \; .
  \end{displaymath}
\end{proposition}

The following proposition is a specialization of
\cite[Proposition~1]{Conradi2020} and \cite[Theorem~2]{osc-011} to a
network where $\cone$ is generated by a single positive vector. It
implies that for such networks to detect a simple Hopf bifurcation one
only needs to study the polynomial~(\ref{cp:bi}) and its Hurwitz
matrices $G_i$, $i=1,2,\ldots , s$.

\begin{proposition}\label{pr:hopf-single-lambda}
Let $\dot{x} = S v(k,x) $ be an ODE system model for  a reaction network
with mass action kinetics where $\rank (S)=s$. Suppose that the matrix
$E$ in (\ref{eq:convex-cone}) consists of a single positive vector and
let $J_{\lambda}( h)$ be the corresponding Jacobian in convex
parameters.
Further let the characteristic
polynomials of $J_{\lambda} (h)$ and $J_{1} (h)$ be as in
(\ref{cp:ai}) and (\ref{cp:bi}), respectively. 
If there exists a fixed value  $h=h^*$ such that 
\begin{equation}
  \begin{split}
    \label{eq:Hurwitz-Ineqs}
    b_s(h^*)&>0 \text{ and} \\
    \det(G_1(h^*)) &> 0,\; \ldots,\; 
    \det(G_{s-2}(h^*)) > 0 \text{ and} \\
    \det(G_{s-1}(h^*))&= 0\ ,
  \end{split}
\end{equation}
then
\begin{enumerate}[{(}a{)}]
\item $J_{1 } (h^*)$ has a single pair of purely imaginary eigenvalues,
\item $J_{\lambda } (h^*)$ has a single pair of purely imaginary
  eigenvalues for all $\lambda >0$,
\item for the dynamical system $\dot{x}= S v(k,x)$
  there exists a simple Hopf bifurcation at $h=h^*$ for all $\lambda >0$
  if there exists some $l \in \{1,2, \ldots , n \}$ such that
  \begin{equation}
    \label{eq:Hur-derivative-G}
    \frac{\partial \det (G_{n-1})}{\partial
      h_l}|_{h_l=h_l^*} \neq 0\ .
  \end{equation}
\end{enumerate}
\end{proposition}

\begin{proof}
  \begin{enumerate}[{(}a{)}]
  \item The proof follows by~\cite[Proposition~1]{Conradi2020}. 
  \item
    By Corollary~\ref{coro:cp_J_lamH}
    $\sign(b_s(h^*))=\sign(a_s(h^*,\lambda))$ for all
    $\lambda>0$. Likewise, by Proposition~\ref{pro:hom}
    $\sign(\det(G_l(h^*)))=\sign(\det(H_l(h^*,\lambda)))$, $l=1$, \ldots, $s$ for
    all $\lambda>0$. Hence by~(a) and
    by~\cite[Proposition~1]{Conradi2020} the Jacobian 
    $J_{\lambda} (h^*)$ has a single pair of purely imaginary
    eigenvalues for all $\lambda >0$.
  \item The fact that $J_{\lambda} (h_l^*)$ for all $\lambda >0$ has a
    single pair of purely imaginary eigenvalues has been established
    in~(b).
    Proposition~\ref{pro:hom} implies the following relationship
    between the two derivatives of
    $\det\left(H_{s-1}(h,\lambda)\right)$ and $\det\left(
      G_{s-1}(h,\lambda) \right)$ with respect to the $h_l$:
    \begin{displaymath}
      \frac{\partial \det (H_{s-1}(h,\lambda))}{\partial h_l}
      =
      \lambda^{\frac{s(s+1)}{2}}
      \frac{\partial \det(G_{s-1}(h))}{\partial  h_l},\, l=1, \ldots, n\ .
    \end{displaymath}
    Thus, if~(\ref{eq:Hur-derivative-G}) holds for
    some $l\in\{ 1$, \ldots, $n\}$, then, for the same~$l$, one has
    \begin{displaymath}
      \frac{\partial \det(H_{s-1}(h,\lambda))}{\partial h_l}|_{(h^*,\lambda)}
      \neq 0 \text{ for all $\lambda>0$.}
    \end{displaymath}
    Therefore, by (a), (b) and by~\cite[Theorem~2\&Remark~2]{osc-011},
    (c) follows as well.
  \end{enumerate}
\end{proof}

\begin{remark}
  The parameter $\lambda$ is not suited as a bifurcation parameter
  in~(\ref{eq:Hur-derivative-G}) as by Proposition~\ref{pro:hom}
  $\det(H_{s-1}(h,\lambda)) = \lambda^{\frac{s(s-1)}{2}}
  \det(G_{s-1}(h))$. Hence $\frac{\partial
    \det(H_{s-1}(h,\lambda))}{\partial \lambda}|_{h=h^*} = 0$, whenever
  $\det(G_{s-1}(h^*))=0$ (cf.\ Remark~\ref{rem:consequences_E_vector}
  explaining that the existence of a pair of purely imaginary
  eigenvalues is independent of $\lambda$).
\end{remark}

\section{Analysis of networks of cyclic distributive double
  phosphorylation}
\label{sec:analysis-network-cyc-simple}

To derive a dynamical model of network~(\ref{eq:MA_cyc}) we assign the
variables given in Table~\ref{tab:species} to the chemical species and
obtain the following ODEs:
\begin{subequations}
  \begin{align}
    \label{eq:sys_cyc_full_1}
    \dot x_{1} &= -\kappa_{1} x_{1} x_{3}- \kappa_{4} x_{1}
                 x_{4}+ (\kappa_{2} +\kappa_{3}) x_{7}+ (\kappa_{5}
                 +\kappa_{6}) x_{8}  \\
    \dot x_{2} &= -\kappa_{10} x_{2} x_{5}-\kappa_{7} x_{2} x_{6}+ (\kappa_{8}
                 +\kappa_{9}) x_{10} +(\kappa_{11}+\kappa_{12}) x_{9}  \\
    \dot x_{3} &= -\kappa_{1} x_{1} x_{3}+ \kappa_{2} x_{7}+\kappa_{12}
                 x_{9}  \\
    \dot x_{4} &= -\kappa_{4} x_{1} x_{4}+\kappa_{3} x_{7}+\kappa_{5}
                 x_{8} \\ 
    \dot x_{5} &= -\kappa_{10} x_{2} x_{5} + \kappa_{9} x_{10}+\kappa_{11}
                 x_{9} \\
    \dot x_{6} &= -\kappa_{7} x_{2} x_{6} + \kappa_{8} x_{10}+\kappa_{6}
                 x_{8} \\ 
    \dot x_{7} &= -(\kappa_{2} +\kappa_{3}) x_{7} + \kappa_{1} x_{1} x_{3} \\
    \dot x_{8} &= -(\kappa_{5} +\kappa_{6}) x_{8} + \kappa_{4} x_{1} x_{4}\\ 
    \dot x_{9} &= -(\kappa_{11} +\kappa_{12}) x_{9} + \kappa_{10} x_{2}
                 x_{5}  \\
    \label{eq:sys_cyc_full_10}
    \dot  x_{10} &= -(\kappa_{8} +\kappa_{9} ) x_{10}+\kappa_{7} x_{2} x_{6} 
  \end{align}
\end{subequations}

Using the same variables $x_1$, \ldots, $x_{10}$ as in
Table~\ref{tab:species}, we obtain the following ODEs:

\begin{subequations}
  \begin{align}
    \label{eq:sys-cyc-1}
    \dot x_{1} &= -\kappa_1 x_{1} x_{3}-\kappa_4 x_{1} x_{4}+\kappa_3
                 x_{7}+\kappa_6 x_{8}, \\
    \dot x_{2} &= \phantom{-}\kappa_9 x_{10}-\kappa_{10} x_{2} x_{5}-\kappa_7
                 x_{2} x_{6}+\kappa_{12} x_{9},  \\
    \dot x_{3} &= -\kappa_1 x_{1} x_{3}+\kappa_{12} x_{9}, \\
    \dot x_{4} &= -\kappa_4 x_{1} x_{4}+\kappa_3 x_{7},  \\
    \dot x_{5} &= \phantom{-}\kappa_9 x_{10}-\kappa_{10} x_{2} x_{5}, \\
    \dot x_{6} &= -\kappa_7 x_{2} x_{6}+\kappa_6 x_{8}, \\
    \dot x_{7} &= \phantom{-}\kappa_1 x_{1} x_{3}-\kappa_3 x_{7}, \\
    \dot x_{8} &= \phantom{-}\kappa_4 x_{1} x_{4}-\kappa_6 x_{8}, \\
    \dot x_{9} &= \phantom{-}\kappa_{10} x_{2} x_{5}-\kappa_{12} x_{9}, \\
    \label{eq:sys-cyc-10}
    \dot x_{10} &= -\kappa_9 x_{10}+\kappa_7 x_{2} x_{6}\ .
  \end{align}
\end{subequations}
Both networks have the same set of three conservation relations, one
for the total amount of  kinase ($c_1$), phosphatase ($c_2$), and
substrate ($c_3$), respectively: 
\begin{align}
  \notag
  x_{1}+x_{7}+x_{8\phantom{0}} &= c_1 \\
  \label{eq:con_rel}
  x_{2}+x_{9}+ x_{10} &= c_2 \\
  \notag
  x_{3}+x_{4}+x_{5}+x_{6}+x_{7}+x_{8}+x_{9}+x_{10} &= c_3
\end{align}

\subsection{Steady states of network~(\ref{net:cyc_simple})}
\label{sec:steady-states-cyc-simple}

For the network~(\ref{net:cyc_simple}) the matrix~$E$ consists of a
single vector of all~1:
\begin{equation}
  \label{eq:E_cyc_d2_simple}
  E^T=(1,1,1,1,1,1,1,1),
\end{equation}
hence network~(\ref{net:cyc_simple}) is a network with a single
extreme ray and we can apply the results of
Section~\ref{sec:single_extreme_vector} and~\ref{sec:Hopf_single_E}.

Given $E$ from~(\ref{eq:E_cyc_d2_simple})
condition~(\ref{eq:diag_k_phi_E_lam}) consequently becomes (cf.\
eq.~(\ref{eq:v_cyc_E}) of Appendix~\ref{app:jac-h-lambda} for
$v(k,x)$):
\begin{displaymath}
  \kappa_1 x_{1} x_{3} = \kappa_3 x_{7} = \kappa_4
  x_{1} x_{4} = \kappa_6 x_{8} = \kappa_7 x_{2}
  x_{6} = \kappa_9 x_{10} = \kappa_{10} x_{2}
  x_{5} = \kappa_{12} x_{9} = \lambda\ .
\end{displaymath}
These equations can be solved for $x$ (in terms of $k$ and
$\lambda$) to obtain the following steady state parameterization:
\begin{equation}
  \begin{split}
    \label{eq:x_lam_k}
    x_{3} &=  \frac{\lambda}{\kappa_1 x_{1}},\, x_{4} =  \frac{\lambda}{\kappa_4
      x_{1}},\, x_{5} =  \frac{\lambda}{\kappa_{10} x_{2}},\, x_{6} =
    \frac{\lambda}{\kappa_7 x_{2}}, \\
    x_{7} &=  \frac{\lambda}{\kappa_3},\, x_{8} =
    \frac{\lambda}{\kappa_6},\, x_{9} =  \frac{\lambda}{\kappa_{12}},\,
    x_{10} = \frac{\lambda}{\kappa_9},
  \end{split}
\end{equation}
where $x_1$ and $x_2$  are arbitrary positive numbers.
Similarly, solving for $k$ one obtains 
\begin{equation}
  \label{eq:ki_hi_lam}
  \begin{split}
    \kappa_1 &= h_1 h_3\lambda ,  \; \kappa_3 =  h_7 \lambda, \; \kappa_4 =
    h_1 h_4\lambda, \;   \kappa_6 =  h_8 \lambda, \\
    \kappa_7 &=  h_2 h_6\lambda,\; \kappa_9 =  h_{10} \lambda, \;   \kappa_{10} =  h_2
    h_5 \lambda, \;   \kappa_{12} = h_9\lambda,
  \end{split}
\end{equation}
where $h_i=\frac{1}{x_i}$, $i=1$, \ldots, $10$ (this
is~(\ref{eq:def_k}) for network~(\ref{net:cyc_simple}) with $v(k,x)$
given in Appendix~\ref{app:jac-h-lambda}).

\subsection{Simple Hopf bifurcations for 
  network~(\ref{net:cyc_simple})}
\label{sec:Hopf}

The Jacobian $J_{\lambda} (h)$ of network~(\ref{net:cyc_simple}) 
computed via (\ref{eq:new-jac}) is 
\begin{equation} \tiny
  \label{eq:Jac_cyc_d2_simple}
  J_{\lambda}(h) = \lambda \,
  \left[
    \begin{array}{rrrrrrrrrr}
      -2 h_{1} & 0 & -h_{3} & -h_{4} & 0 & 0 & h_{7} & h_{8} & 0 & 0 \\
      0 & -2 h_{2} & 0 & 0 & -h_{5} & -h_{6} & 0 & 0 & h_{9} & h_{10} \\
      -h_{1} & 0 & -h_{3} & 0 & 0 & 0 & 0 & 0 & h_{9} & 0 \\
      -h_{1} & 0 & 0 & -h_{4} & 0 & 0 & h_{7} & 0 & 0 & 0 \\
      0 & -h_{2} & 0 & 0 & -h_{5} & 0 & 0 & 0 & 0 & h_{10} \\
      0 & -h_{2} & 0 & 0 & 0 & -h_{6} & 0 & h_{8} & 0 & 0 \\
      h_{1} & 0 & h_{3} & 0 & 0 & 0 & -h_{7} & 0 & 0 & 0 \\
      h_{1} & 0 & 0 & h_{4} & 0 & 0 & 0 & -h_{8} & 0 & 0 \\
      0 & h_{2} & 0 & 0 & h_{5} & 0 & 0 & 0 & -h_{9} & 0 \\
      0 & h_{2} & 0 & 0 & 0 & h_{6} & 0 & 0 & 0 & -h_{10} \\
    \end{array}
  \right]\ .
\end{equation}
For $J_{\lambda}(h)$ given in (\ref{eq:Jac_cyc_d2_simple}) one has
$\rank(J_{\lambda}(h))=7$. Hence its characteristic polynomial is of
the following form  (cf.~supplementary file
\texttt{Cyc\_dd\_2\_coeffs\_charpoly.nb} and 
Corollary~\ref{coro:cp_J_lamH}):
\begin{equation}\label{eq:cp-cyc-ai}
  \det \left(\mu I - J_{\lambda }(h) \right)
  ~=~ \mu^3 \left(
    \mu^7 + \lambda b_1(h) \mu^{6} + \ldots + \lambda^{6} b_{6}(h) \mu +
    \lambda^7 b_7(h)
  \right),
\end{equation}
where the coefficients $b_1(h)$, \ldots, $b_7(h)$ are the coefficients
of the characteristic polynomial of the matrix $J_1(h)$.
With the next corollary we adapt
Proposition~\ref{pr:hopf-single-lambda} to the Jacobian of
network~(\ref{net:cyc_simple}). This allows us to work with the
simpler coefficients $b_i(h)$.

\begin{corollary}\label{cor:ds-g-hopf}
  Consider the dynamical systm~(\ref{eq:sys-cyc-1}) --
  (\ref{eq:sys-cyc-10}) defined by 
  network~(\ref{net:cyc_simple}) with Jacobian 
   $J_{\lambda} (h)$ as in~(\ref{eq:Jac_cyc_d2_simple}).
  Consider the polynomial~(\ref{eq:cp-cyc-ai}) for $\lambda=1$ and
  obtain its coefficients $b_1(h), \ldots, $ $b_7(h)$ and
  Hurwitz-matrices $G_1(h)$, \ldots, $G_6(h)$.
  Assume that at some $h=h^*$ the following conditions hold:
  \begin{equation}
    \begin{split}
      \label{eq:Hurwitz-Ineqs}
      b_7(h^*)&>0 \text{ and} \\
      \det(G_1(h^*)) &> 0,\; \ldots,\; 
      \det(G_5(h^*)) > 0 \text{ and} \\
      \det(G_6(h^*))&= 0\ .
    \end{split}
  \end{equation}
  Then, for all $\lambda>0$, the dynamical system (\ref{eq:sys-cyc-1})
  -- (\ref{eq:sys-cyc-10}) has
  a simple Hopf bifurcation at~$h=h^*$, if additionally 
  \begin{equation}
    \label{eq:Hurwitz-derivative}
    \exists l\in\{1, \ldots, 10\} \text{ such that } \frac{\partial
      \det(G_6)}{\partial h_l}|_{h=h^*} \neq 0 \quad  \ .
  \end{equation}
\end{corollary}

The coefficients $b_0 (h)$, \ldots, $b_5 (h)$ and $b_7 (h)$, as well
as the Hurwitz determinants $\det(G_2 (h))$, \ldots, $\det(G_5 (h))$
contain only monomials with positive sign (cf.~supplementary file
\texttt{Cyc\_dd\_2\_coeffs\_charpoly.nb}). 
This establishes the following Lemma and Corollary:
\begin{lemma}
  \label{lem:sum_of_posinomials}
  Consider the coefficients of the characteristic polynomial
  (\ref{eq:cp-cyc-ai}) of $J_{\lambda}(h)$ given in
  (\ref{eq:Jac_cyc_d2_simple}) for $\lambda=1$ and its Hurwitz
  determinants. For all $h>0$ the following holds:
  \begin{enumerate}[{(}A{)}]
  \item $b_0(h) >0$, \ldots, $b_5(h)>0$ and
    $b_7(h)>0$ and 
  \item $\det(G_2 (h))>0$, \ldots, $\det(G_5 (h))>0$.
  \end{enumerate}
\end{lemma}

The Hurwitz determinant $\det(G_6(h))$ contains monomials of both 
signs (cf.~supplementary file \texttt{cyc\_dd\_2\_detH6.txt}) , hence it
can potentially be zero. To establish this, we 
consider $\det(G_6(h))$ as a polynomial in $h_1$, $h_2$, $h_3$ and
$h_6$ only and study its Newton polytope. Using \texttt{polymake}
\cite{Assarf2017,Gawrilow2000} we compute the following hyperplane
representation of this Newton polytope:
\begin{align*}
  -h_1 &\geq -6
  & -h_1 - h_2 &\geq -11
  & -h_1 - h_2 - h_3 &\geq -15 \\
  -h_1 - h_2 - h_3 - h_6 &\geq -18
  &  h_6 &\geq 0
  & -h_2 &\geq -6 \\
  -h_1 - h_3 &\geq -11
  & -h_3 - h_6 &\geq -11
  & -h_2 - h_3 - h_6 &\geq -15 \\
  -h_1 - h_3 - h_6 &\geq -15
  & h_1 &\geq 0
  & -h_2 - h_3 &\geq -11\\
  -h_1 - h_2 - h_6 &\geq -15
  & -h_3 &\geq -6
  & -h_2 - h_6 &\geq -11 \\
  h_2 &\geq 0
  & -h_1 - h_6 &\geq -11 & -h_6 &\geq -6 \\
  h_3 &\geq 0
\end{align*}
We study the coefficients of $\det(G_6)$ as a polynomial
in $h_1$, $h_2$, $h_3$ and $h_6$ and find that some of these contain
the factor 
\begin{displaymath}
  h_{10} h_{7}-h_{8} h_{9}\ .
\end{displaymath}
In fact, visual inspection of all coefficients shows that 
all monomials with exponent vectors contained in the hyperplane  
\begin{displaymath}
    h_1 + h_2 + h_3 +h_6 = 18
\end{displaymath}
have such a coefficient that factors   $h_{10} h_{7}-h_{8} h_{9}$.
Hence we want to make those monomials dominant. To achieve
this we use the following transformation (based on the normal vector
of the hyperplane):
\begin{equation}
  \label{eq:trafo_supp_hyp}
  h_{1}\to t,h_{2}\to t,h_{3}\to t,h_{6}\to t.
\end{equation}
The result is the following degree 18 polynomial in $t$, with
coefficients that are polynomials in $h_4$, $h_5$  and $h_7$, \ldots,
$h_{10}$:  
\begin{displaymath}
  \begin{split}
    D_6(t) &= 324 t^{18} (h_{10}+h_{7}) (h_{10} h_{7}-h_{8}
    h_{9}) + \ldots \\
    &\quad
    + h_{10} h_{4} h_{5} h_{7} h_{8} h_{9} \\
    &\qquad 
    \cdot (h_{10}+h_{4})
    (h_{10}+h_{5}) (h_{10}+h_{7}) (h_{10}+h_{8}) (h_{10}+h_{9}) \\
    & \qquad 
    \cdot (h_{4}+h_{5}) (h_{4}+h_{7}) (h_{4}+h_{8}) (h_{4}+h_{9}) \\
    &\qquad 
    \cdot (h_{5}+h_{7}) (h_{5}+h_{8}) (h_{5}+h_{9}) \\
    & \qquad
    \cdot (h_{7}+h_{8}) (h_{7}+h_{9}) \\
    &\qquad
    \cdot (h_{8}+h_{9})\ . 
  \end{split}
\end{displaymath}
As the constant coefficient is a sum of positive monomials one has
$D_6(0)>0$. Thus, if
\begin{displaymath}
  h_{10} h_{7}-h_{8} h_{9} < 0,
\end{displaymath}
then there exists a $t_1>0$ such that $D_6(t_1)=0$  and $D_6 (t)<0$ for
$t>t_1$ by the Intermediate Value Theorem.
These observations are the basis for the following result:
\begin{lemma}
  \label{lem:existence_hopf}
  Consider the coefficients of the characteristic polynomial given in
  (\ref{eq:cp-cyc-ai}) for $\lambda=1$  and obtain its
  coefficients $b_1(h), \ldots, $ $b_7(h)$ and its Hurwitz-matrices
  $G_1(h)$, \ldots, $G_6(h)$.
  Let $h(t)$ be such that
  \begin{equation}
    \label{eq:h_main}
    h(t) = \left( t, t, t, h_4, h_5, t, h_7, h_8, h_9, h_{10} \right)^T
  \end{equation}
  with
  \begin{equation}
    \label{eq:h_ineq}
    h_{7}h_{10}-h_{8} h_{9} < 0\ .
  \end{equation}
  Let $b_i(h(t))= b_i (t)$ and $D_i (t)  =\det G_i (h(t))$, $i=1$,
  \ldots, $6$. Then there exists a positive real number $t_1$,  such
  that
  \begin{displaymath}
    D_6(t) < 0  \text{ for } t >t_1 \text{ and }
    D_6(t_1) = 0\ .
  \end{displaymath}
   In addition, $b_7(t_1)>0$  and $D_i (t_1) >0$ for $i=1,2, \ldots , 5$.
\end{lemma}

\begin{proof}
  Choose positive values $h_4^*$, $h_5^*$ and positive
  values $h_7^*$, $h_8^*$, $h_9^*$ and $h_{10}^*$ that
  satisfy~(\ref{eq:h_ineq}). Fix $h_4=h_4^*$, $h_5=h_5^*$ and
  $h_7=h_7^*$, \ldots, $h_{10} = h_{10}^*$ to obtain $h^*(t)$ (which
  now depends only on $t$). Evaluate the $b_i$'s and $\det(G_i)$ at
  $h^*(t)$ to obtain the $t$-polynomials $b_i(t)\equiv
  b_i(h^*(t))$ and
  $D_i(t)\equiv\det(G_i(h^*(t)))$. 
  
  The existence of $t_1$ with $D_6(t)<0$, $t>t_1$ and $D_6(t_1)=0$ has
  been established above. By
  Lemma~\ref{lem:sum_of_posinomials} $b_1$, \ldots $b_5$ and $b_7$ as
  well as 
  $\det(G_2)$, \ldots, $\det(G_5)$ are sums of positive 
  monomials and thus, in particular, positive if evaluated at
  $h^*(t_1)$. Thus $b_7(t_1)$ and  $D_i(t_1)$,  $i=1,2, \ldots ,5$
  are positive. 
\end{proof}

\begin{theorem}\label{thm:stab}
  Consider the dynamical system~(\ref{eq:sys-cyc-1}) --
  (\ref{eq:sys-cyc-10}) arising from
  network~(\ref{net:cyc_simple}) with Jacobian $J_{\lambda}(h)$
  as in (\ref{eq:Jac_cyc_d2_simple}). Let $h(t)$ be as
  in~(\ref{eq:h_main}) and $\lambda>0$. Fix the remaining $h_i=h_i^*$
  such that the inequality~(\ref{eq:h_ineq}) is satisfied. Then for
  $t=t_1$ computed as in Lemma~\ref{lem:existence_hopf}, the
  dynamical system~(\ref{eq:sys-cyc-1}) -- (\ref{eq:sys-cyc-10})
  undergoes a simple Hopf bifurcation for  
  all $\lambda >0$ if
  \begin{displaymath}
    \frac{d \det(G_6(t))}{d t}|_{t=t_1} \neq 0\ .
  \end{displaymath}
\end{theorem}
\begin{proof}
  If $t=t_1$ is  as in Lemma~\ref{lem:existence_hopf}, the
  conditions~(\ref{eq:Hurwitz-Ineqs})
  of Corollary~\ref{cor:ds-g-hopf}  are satisfied.
  If $\frac{d \det(G_6)}{d t}|_{t=t_1} \neq 0$, then at least one of
  the derivatives $\frac{\partial \det(G_6)}{\partial h_i}|_{h=h(t_1)}
  \neq 0$ (by the chain rule) and hence condition
  (\ref{eq:Hurwitz-derivative}) is satisfied. Thus, by
  Corollary~\ref{cor:ds-g-hopf}, the dynamical
  system~(\ref{eq:MA_cyc}) undergoes a Hopf bifurcation at $t=t_1$ in
  this case.
\end{proof}

Theorem~\ref{thm:stab} establishes the existence of simple Hopf
bifurcations for the the dynamical system~(\ref{eq:sys-cyc-1}) --
(\ref{eq:sys-cyc-10}) derived from network~(\ref{net:cyc_simple}).
In the following remark we observe that inequality
(\ref{eq:ineq_multi_hopf}) and inequality~(\ref{eq:h_ineq}) are
equivalent:
\begin{remark}
  First recall that the $h_i$ have to satisfy (\ref{eq:ki_hi_lam})
  (as ($\kappa$,$\frac{1}{h}$) is a solution to the steady state 
  equation~(\ref{eq:def_ss})). That is, $h_7$, \ldots, $h_{10}$ can
  be represented in terms of $\kappa_3$, $\kappa_6$, $\kappa_9$ and
  $\kappa_{12}$ (and $\lambda$) as  
  \begin{displaymath}
    h_{7} =  \frac{\kappa_3}{\lambda},
    h_{8} =  \frac{\kappa_6}{\lambda},
    h_{9} =  \frac{\kappa_{12}}{\lambda},
    h_{10} = \frac{\kappa_9}{\lambda}\ .
  \end{displaymath}
  Using this one obtains for the left hand side of
  inequality~(\ref{eq:h_ineq}):
  \begin{displaymath}
    h_7 h_{10} - h_8 h_9 = \frac{1}{\lambda^2}\left( \kappa_3 \kappa_9 - \kappa_6
      \kappa_{12} \right)\ . 
  \end{displaymath}
  As $\lambda^2>0$ the inequality (\ref{eq:h_ineq}) is equivalent
  \begin{displaymath}
    \kappa_3 \kappa_9 - \kappa_6  \kappa_{12} < 0\ .
  \end{displaymath}
\end{remark}

Finally, we
make several remarks  regarding the stability of the
positive steady state $x^*(t)=\frac{1}{h(t)}$ of the
system~(\ref{eq:sys-cyc-1}) -- (\ref{eq:sys-cyc-10})
depending on $t>0$.
 
\begin{remark}\label{rem:stability}
  \begin{enumerate}[{(}a{)}]
  \item
    For sufficiently small $t>0$, the
    positive steady state $x^*(t)$ is asymptotically stable. (This is
    true regardless of inequality~(\ref{eq:h_ineq})).
  \item
    Suppose that  the inequality  (\ref{eq:h_ineq}) is
    satisfied. Then for sufficiently  large $t >0$,  $x^*(t)$ is
    unstable. 
  \end{enumerate}
\end{remark}

\subsection{A procedure to locate simple Hopf bifurcations in
  network~(\ref{net:cyc_simple})}
\label{sec:procedure}

Following the proof of Lemma~\ref{lem:existence_hopf} we proceed as
follows to find points that satisfy (\ref{eq:Hurwitz-Ineqs}):
\begin{enumerate}[{Step~}1{:}]
\item\label{item:Step1} Choose positive values for $h_{7}$, $h_{8}$,
  $h_{9}$, $h_{10}$ such that (\ref{eq:h_ineq}) is satisfied (e.g.\
  $h_{7}=h_{8}=h_{10}=1$ and $h_{9}=2$).
\item Choose positive values for $h_4$ and $h_5$ (e.g. $h_4 = h_5 =1$).
\item In $p(t)$ set $h_{1}=h_{2}=h_{3}=h_{6}= t$. For the values
  chosen so far we obtain
  \begin{displaymath}
    \begin{split}
      & 497664 + 10946304 t + 103721056 t^2 + 579850652 t^3 + 
      2169242876 t^4 \\
      &
      + 5787611019 t^5 + 11398671182 t^6 + 
      16865933820 t^7 + 18863357157 t^8  \\
      &
      + 15900121640 t^9 + 
      9989687485 t^{10} + 4589099030 t^{11} + 1497364081 t^{12}  \\
      &
      + 
      331280824 t^{13} + 45135703 t^{14} + 2794428 t^{15} \\
      & - 85122
      t^{16} - 20304 t^{17} - 648 t^{18} = 0
    \end{split}
  \end{displaymath}
\item Approximate the positive real root(s) of $p(t)=0$. For the values
  chosen so far we obtain $t^*\approx 14.874$.
\item\label{item:before_last}
  Check that the point generated so far
  satisfies~(\ref{eq:Hurwitz-derivative}). In the example we obtain 
  \begin{displaymath}
    h^*=(14.874,14.874,14.874,1,1,14.874,1,1,2,1)^T. 
  \end{displaymath}
  We use \texttt{Matcont} to verify a Hopf bifurcation,
  cf.~Fig.~\ref{fig:S11_vs_c1}. 
\item\label{item:last}
  Choose some $t>t^*$, compute vectors $h$ and
  $x=\frac{1}{h}$ and use eq.~(\ref{eq:def_k}) and~(\ref{eq:con_rel})
  to obtain rate constants and total concentrations.
  We have chosen $t=15$ and hence obtain
  \begin{equation}
    \begin{split}
      \label{eq:h_x0_exa}
      &h^T =(15,15,15,1,1,15,1,1,2,1) \text{ and } \\
      &x^T =
    \left(\frac{1}{15},\frac{1}{15},\frac{1}{15},1,1,\frac{1}{15},1,1,\frac{1}{2},
      1 \right)
  \end{split}
\end{equation}
  and the rate constants and total concentrations given in
  Table~\ref{tab:rc_cyc_exa}.
  \begin{table}[!h]
    \centering
    \begin{tabular}{|c|c|c|c|c|c|c|c||c|c|c|} \hline
      $\kappa_1$ & $\kappa_3$ & $\kappa_4$ & $\kappa_6$ & $\kappa_7$
      & $\kappa_9$ & $\kappa_{10}$ & $\kappa_{12}$ & $c_1$ & $c_2$ & $c_3$ \\ \hline 
      225 $\lambda$ & $\lambda$ & 15 $\lambda$ & $\lambda$ & 225 $\lambda$
      & $\lambda$ & 15 $\lambda$ & 2 $\lambda$ & $\frac{31}{15}$
      & $\frac{47}{30}$ & $\frac{169}{30}$ \\ \hline
    \end{tabular}
    \caption{
      \label{tab:rc_cyc_exa}
      Rate constants and total concentrations obtained
      by solving (\ref{eq:diag_k_phi_E_lam}) for~$k$.
    }
  \end{table}
\end{enumerate}

\begin{remark}
  \label{rem:candidate-Hopf-points}
  As a consequence of Lemma~\ref{lem:existence_hopf}, any point $x^*$
  obtained via Steps~\ref{item:Step1} to~\ref{item:before_last} is a {\em
    candidate} Hopf point. It is guaranteed to
  satisfy~(\ref{eq:Hurwitz-Ineqs}), but a simple Hopf bifurcation only
  occurs if the condition~(\ref{eq:Hurwitz-derivative}) is satisfied as
  well. In practice we suggest the following approach: first determine
  a candidate point $x^*$ via Steps~\ref{item:Step1}
  to~\ref{item:before_last}, second use (\ref{eq:ki_hi_lam}) and
  (\ref{eq:con_rel}) to determine the corresponding rate constants and
  total concentrations and third verify the existence of a simple Hopf
  bifurcation by using a numerical continuation software like MATCONT
  \cite{Kuznetsov2003} to vary a rate constant or total concentration
  at this point. 
\end{remark}

\begin{figure}[!h]
  \centering
  \begin{subfigure}{0.3\textwidth}
    \includegraphics[width=0.9\linewidth]{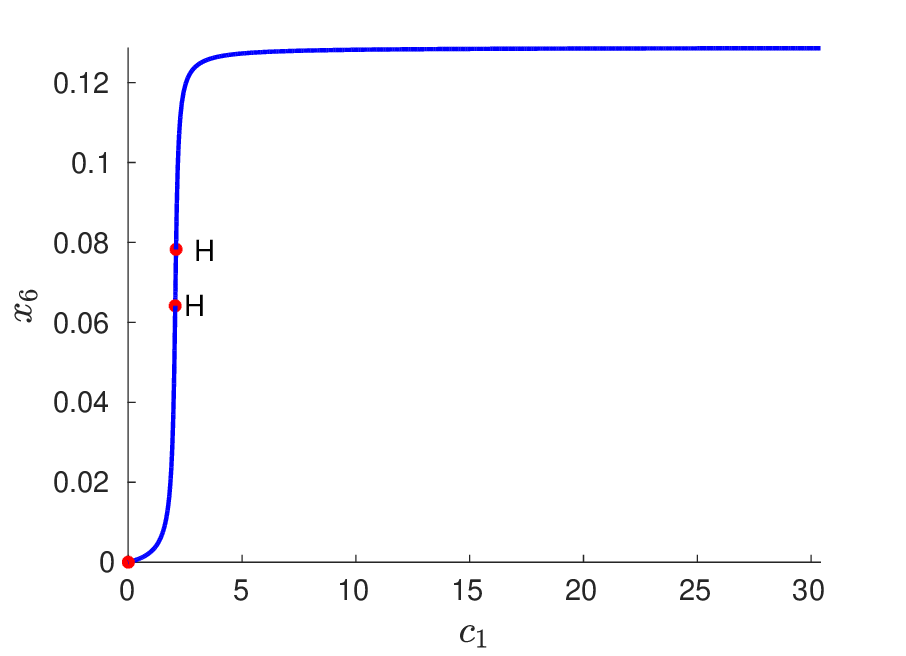}
    \subcaption{
      \label{fig:S11_vs_c1}
      $x_6$ vs.\ $c_1$
    }
  \end{subfigure}
  \begin{subfigure}{0.3\textwidth}
    \includegraphics[width=0.9\linewidth]{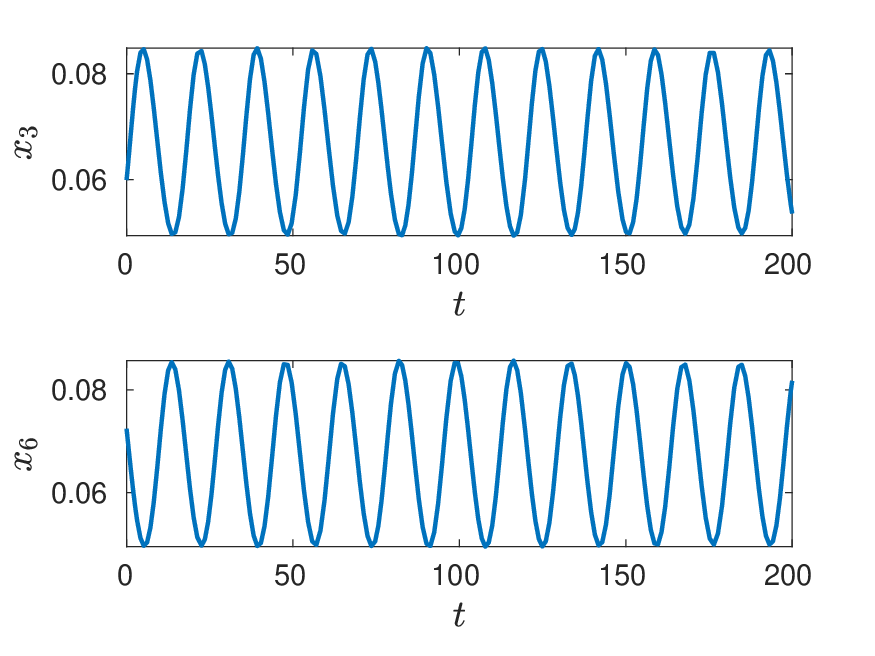}
    \subcaption{
      \label{fig:Sxx_vs_t}
      $x_3$ and $x_6$ vs.\ $t$
    }
  \end{subfigure}
  \begin{subfigure}{0.3\textwidth}
    \includegraphics[width=0.9\linewidth]{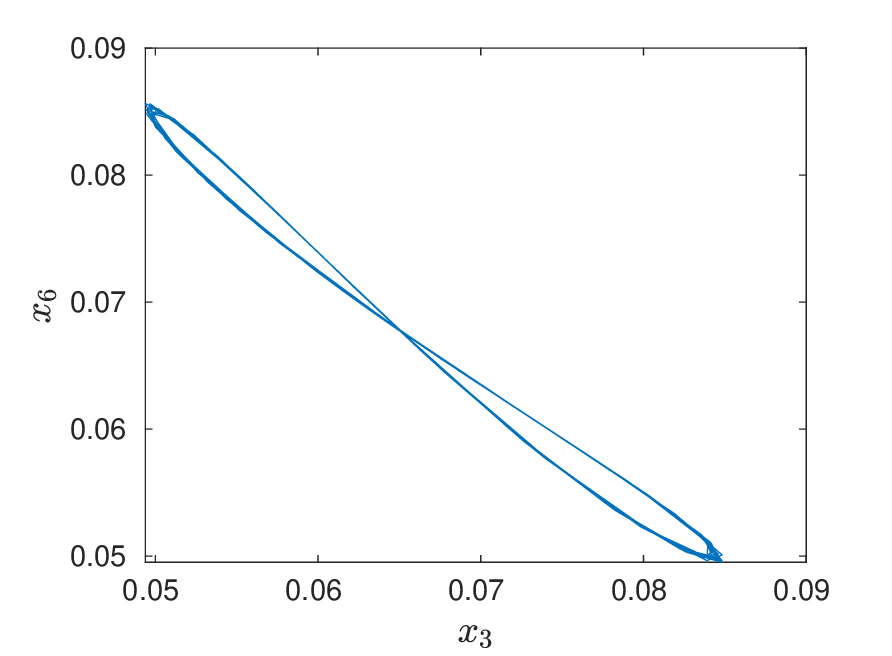}
    \subcaption{
      \label{fig:S11_vs_S00}
      $x_6$ vs.\ $x_3$
    }
  \end{subfigure}
  
  \caption{
    \label{fig:LC_small}
    Numerical verification of Hopf bifurcations
    (panel~(\subref{fig:S11_vs_c1}), labeled~$H$) and a limit 
    cycle (panel~(\subref{fig:Sxx_vs_t})
    and~(\subref{fig:S11_vs_S00})). Rate constants~$k$ as in
    Table~\ref{tab:rc_cyc_exa} with $\lambda=1$. Initial value
    $x(0)$ as in~(\ref{eq:h_x0_exa}) -- apart from $x_3(0)$ and
    $x_6(0)$: to obtain an initial value near the steady state $x$
    given in~(\ref{eq:h_x0_exa}) we choose $x_3(0) =
    1.1\cdot\frac{1}{15}$ and $x_6(0)=0.9\cdot\frac{1}{15}$).
  } 
\end{figure}

\begin{remark}
  \label{rem:cyc_k_fixed}
  In the course of Steps~\ref{item:Step1} -- \ref{item:last} above all rate
  constants are fixed:
  \begin{enumerate}[{(}i{)}]
  \item From~(\ref{eq:def_k}) one obtains for
    network~(\ref{net:cyc_simple}) the relation~(\ref{eq:ki_hi_lam})
    between the~$\kappa_i$ and the~$h_i$.
  \item Thus, choosing numerical values for $h_7$, \ldots, $h_{10}$ is
    equivalent to choosing $\kappa_3$, $\kappa_6$, $\kappa_9$ and $\kappa_{12}$ (up to the
    factor $\lambda$).
  \item Choosing numerical values for $h_4$ and $h_5$ and assigning
    $h_1=h_2=h_3=h_6=t$ is equivalent to choosing $\kappa_1$, $\kappa_4$, $\kappa_7$
    and $\kappa_{10}$ (again up to the factor $\lambda$). In this case $\kappa_1$
    and $\kappa_7$ are proportional to $t^2$ and $\kappa_4$ and $\kappa_{10}$ to $t$ --
    as $h_4$ and $h_5$ are fixed to numerical values. 
  \item For $h_7$, \ldots, $h_{10}$ chosen in Step~\ref{item:Step1}
    and $\lambda=1$ one thus obtains 
    $k=\left(t^2, 1, t, 1, t^2, 1, t, 2
    \right)$.
  \end{enumerate}
\end{remark}
 
\subsection{Lifting to the full network~(\ref{eq:MA_cyc})
\label{sec:lifting}}

In \cite{Banaji2018} network modifications are described that preserve the
existence of a stable positive limit cycle. This is the basis for
the following result:
\begin{theorem} 
  \label{theo:main-cyc-full}
  Consider networks~(\ref{eq:MA_cyc})
  and~(\ref{net:cyc_simple}). Assume 
  the rate constant values 
  of network~(\ref{net:cyc_simple}) are such that the
  system~(\ref{eq:sys-cyc-1}) -- (\ref{eq:sys-cyc-10}) admits a stable
  limit cycle.
  Choose these values for the rate constants of
  network~(\ref{eq:MA_cyc}).
  For values of $\kappa_2$, $\kappa_5$, $\kappa_8$, $\kappa_{11}$ 
  small enough, there exists a stable limit cycle in the  
  system~(\ref{eq:sys_cyc_full_1}) -- (\ref{eq:sys_cyc_full_10}) close
  to the limit cycle of the system~(\ref{eq:sys-cyc-1}) --
  (\ref{eq:sys-cyc-10}).
\end{theorem}

\begin{proof}
In the language of \cite{Banaji2018}, if, a network that admits a
stable positive limit cycle (for some values of the rate constants
and initial conditions) is modified by adding reactions that are in
the span of the stoichiometric matrix, then the new network also
admits a stable positive limit cycle~\cite[Theorem~1]{Banaji2018}
(if the rate constants of the new reactions are chosen
appropriately). 

As the reaction vectors of reversible reactions are in the span of the
stoichiometric matrix, the existence of a limit cycle in
network~(\ref{net:cyc_simple}) implies the existence of a limit cycle
in the full network~(\ref{eq:MA_cyc}) for
appropriately chosen rate constants of the backward reactions.
\end{proof}

To illustrate Theorem~\ref{theo:main-cyc-full}, we use the ODEs 
(\ref{eq:sys_cyc_full_1}) -- (\ref{eq:sys_cyc_full_10}) and the values
of Table~\ref{tab:rc_cyc_exa} (on page~\pageref{tab:rc_cyc_exa}).
Fig.~\ref{fig:lifting} demonstrates
the existence of a limit cycle in the full system
(\ref{eq:sys_cyc_full_1}) -- (\ref{eq:sys_cyc_full_10})(for  $k_b$
sufficiently small).

\begin{remark}[Locating the limit cycle in Fig.~\ref{fig:small_kb}]
  For simplicity we choose $\kappa_2 = \kappa_5 = \kappa_8 =
  \kappa_{11} = k_b$.
  To obtain Fig.~\ref{fig:lifting} the initial value given in
  the table of Fig.~\ref{tab:ini_on_orbit} was used. This point is \lq
  close\rq{} to the limit cycle of the ODEs defined by 
  network~(\ref{net:cyc_simple}) (i.e.\ for $k_b=0$). It was obtained by 
  solving the ODEs~(\ref{eq:sys_cyc_full_1}) --
  (\ref{eq:sys_cyc_full_10}) using Matlab's \texttt{ode15s} with
  initial value given in Table~\ref{tab:rc_cyc_exa} (on
  page~\pageref{tab:rc_cyc_exa}) for a \lq long\rq{} time (i.e.\ until
  $T=5000$) and $\lambda=1$. The point in the table of
  Fig.~\ref{tab:ini_on_orbit} corresponds to the last point of that
  first simulation.
\end{remark}

\begin{figure}[!h]
  \centering
  \begin{subfigure}{0.45\textwidth}
    \includegraphics[width=0.95\textwidth]{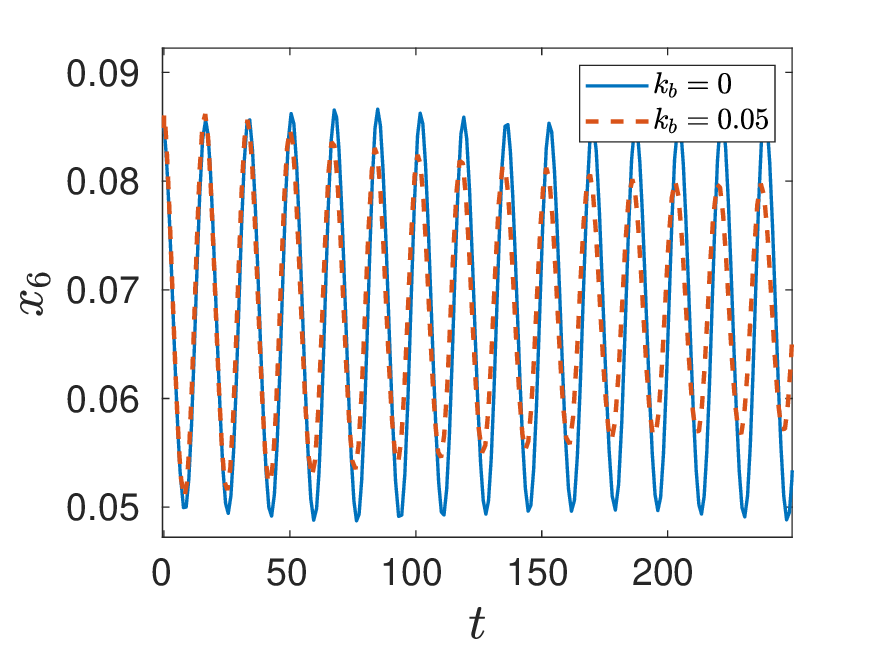}
    \subcaption{
      \label{fig:small_kb}
       $x_6$ vs.\ $t$ for $k_b=0$, $0.05$.
    }
  \end{subfigure}
  \hfill
  \begin{subfigure}{0.45\textwidth}
    \includegraphics[width=0.95\textwidth]{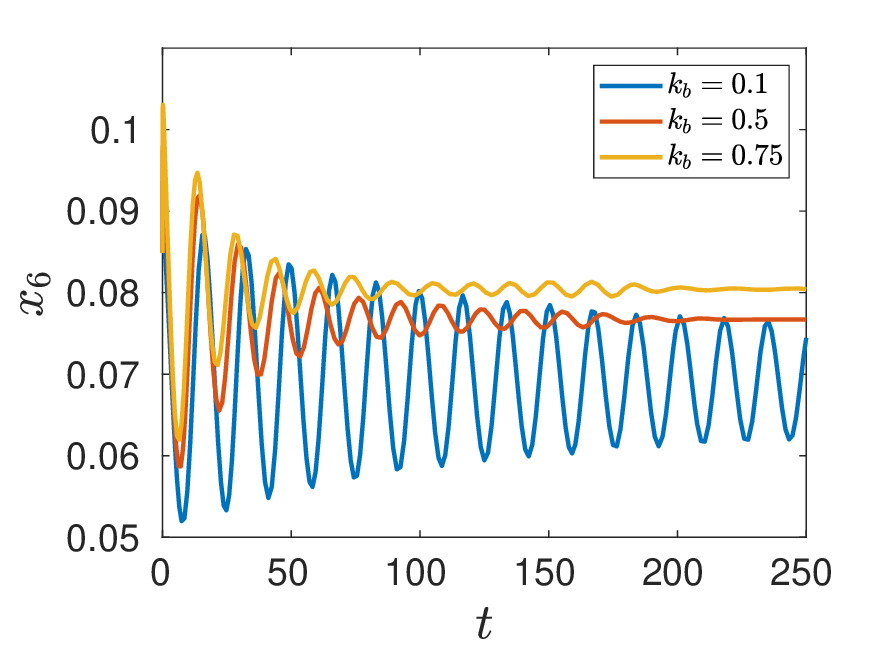}
    \subcaption{
      \label{fig:larger_kb}
      $x_6$ vs.\ $t$ for $k_b=0.1$, $0.5$, $0.75$.
    }
  \end{subfigure}

  \vspace{2ex}
  \begin{subfigure}{0.95\textwidth}
    \begin{center}
      \tiny
      \begin{tabular}{|c|c|c|c|c|c|c|c|c|c|}\hline
        $x_1$ & $x_2$ &  $x_3$ & $x_4$ & $x_5$ & $x_6$ & $x_7$
        & $x_8$ & $x_9$ & $x_{10}$ \\ \hline 
        0.0885267 & 0.0528367 & 0.0496013 & 0.74587 &1.261 & 0.084892
        & 0.97124 & 1.0069 & 0.49383 & 1.02 \\ \hline
      \end{tabular}
      \subcaption{
        \label{tab:ini_on_orbit}
        Initial value for simulations displayed in
        Fig.~\ref{fig:lifting}.
      }
    \end{center}
  \end{subfigure}
  
  \caption{\label{fig:lifting} Simulation of network~(\ref{eq:MA_cyc})
    for $\kappa_2 = \kappa_5 = \kappa_8 = \kappa_{11} = k_b$ and
    different values~$k_b$ (ODEs have been solved with \texttt{ode15s}
    (Mathworks) for $x(0)$  as in the  table of
    panel~(\subref{tab:ini_on_orbit}) and  $\kappa_i$, $c_i$ as in  
    Table~\ref{tab:rc_cyc_exa} (on page~\pageref{tab:rc_cyc_exa} with
    $\lambda=1$). Panel~(\subref{fig:small_kb}): $k_b=0$ corresponds 
    to the ODEs derived from  network~(\ref{net:cyc_simple}),
    $k_b=0.05$ to the ODES~(\ref{eq:sys_cyc_full_1}) --
    (\ref{eq:sys_cyc_full_10})
    for $k_b=0.05$. The oscillations indicate 
    for $k_b=0.05$ a stable limit cycle close to the stable
    limit cycle for $k_b=0$. Panel~(\subref{fig:larger_kb}): the
    stable limit cycle does not seem to exist for larger values
    of~$k_b$.
  }
\end{figure}

\newpage
\section{Discussion}
\label{sec:discussion}

In this section we discuss inequality~(\ref{eq:ineq_multi_hopf})
in the light of results on multistationarity for a network of
sequential distributive double phosphorylation described in
\cite{maya-bistab}. In Section~\ref{sec:cyc_vs_seq} we introduce
the corresponding reaction network and compare it to
network~(\ref{eq:MA_cyc}). In Section~\ref{sec:cycl-dist-osci} we
briefly summarize the results presented in
Section~\ref{sec:analysis-network-cyc-simple} and in 
Section~\ref{sec:sequ-distr-bistab} the multistationarity results of
\cite{maya-bistab}. We close by arguing our conclusion that in
distributive double phosphorylation the catalytic constants enable
non-trivial dynamics in Section~\ref{sec:cat-const-non-triv}.

\subsection{Cyclic versus sequential distributive double
  phosphorylation}
\label{sec:cyc_vs_seq}

Sequential and distributive double phosphorylation can be described
by the following mass action network (cf.\ e.g.\ \cite{Holstein2013}
or~\cite{maya-bistab}): 
\begin{equation}
  \label{eq:MA-seq}
  \begin{split}
    S_0 + K \ce{<=>[\kappa_1][\kappa_2]} KS_0 \ce{->[\kappa_3]} S_1+K
    \ce{<=>[\kappa_4][\kappa_4]} KS_1 \ce{->[\kappa_6]} S_2+K \\
    S_2 + F  \ce{<=>[\kappa_7][\kappa_8]} FS_2 \ce{->[\kappa_9]} S_1+F
    \ce{<=>[\kappa_{10}][\kappa_{11}]} FS_1 \ce{->[\kappa_{12}]} S_0+F.
  \end{split}
\end{equation}
Network~(\ref{eq:MA-seq}) is structurally similar to
network~(\ref{eq:MA_cyc}): both networks contain 12 reactions and
the only difference is that network~(\ref{eq:MA_cyc}) contains
two species of mono-phosphorylated protein ($S_{10}$ and $S_{01}$),
while network(\ref{eq:MA-seq}) contains only one ($S_1$). Hence
network~(\ref{eq:MA-seq}) contains nine species, while
network network~(\ref{eq:MA_cyc}) contains ten.

In particular, networks~(\ref{eq:MA_cyc}) and~(\ref{eq:MA-seq})
contain the same four phosphorylation events: (i) the conversion of
unphosphorylated protein to mono-phosphorylated protein catalyzed be
the kinase $K$, (ii) the conversion of mono-phosphorylated protein
to double-phosphorylated protein catalyzed by the same
kinase $K$, (iii) the conversion of double-phosphorylated protein to
mono-phosphorylated protein catalyzed by the phosphatase $F$ and
(iv) the conversion of mono-phosphorylated protein to
unphosphorylated protein catalyzed by the same phosphatase $F$. As 
described in~\cite{maya-bistab}, in enzyme kinetics research it is
customary to characterize such  phosphorylation events by three
constants, the Michaelis constant~($K_m$), the catalytic
constant~($k_c$) and the equilibrium constant $k_{eq}$ of the
respective enzyme substrate pair (see, for
example,~\cite{Bowden2004} for details on enzyme kinetics). 

Of particular interest in the context of the present publication
are the $k_c$-values as these correspond to the rate constants
involved in inequality~(\ref{eq:ineq_multi_hopf}): $\kappa_3$ is the
$k_c$-value of the kinase $K$ with unphosphorylated substrate
($S_{00}$ or $S_0$), $\kappa_6$ of $K$ with mono-phosphorylated 
substrate ($S_{10}$ or $S_1$), $\kappa_9$ of $F$ with
double-phosphorylated substrate ($S_{11}$ or $S_2$) and
$\kappa_{12}$ of $F$ with mono-phosphorylated substrate ($S_{10}$ or
$S_1$).

\subsection{Cyclic and distributive: emergence of oscillations} 
\label{sec:cycl-dist-osci}

By Theorem~\ref{thm:stab}, if these catalytic constants satisfy
inequality~(\ref{eq:ineq_multi_hopf}), then there exists positive
steady states of network~(\ref{net:cyc_simple}) such that the
Jacobian has a complex-conjugate pair of eigenvalues on the
imaginary axis. This is necessary for a simple Hopf-bifurcation. If
there is a supercritical simple Hopf bifurcation and a stable limit
cycle emerges, then by Theorem~\ref{theo:main-cyc-full} there is a
stable limit cycle in network~(\ref{eq:MA_cyc}). Hence we say that
for cyclic and distributive double phosphorylation the catalytic
constants enable the emergence of oscillations.

\subsection{Sequential and distributive: emergence of bistability}
\label{sec:sequ-distr-bistab}

In \cite{maya-bistab} we have shown that
the  inequality~(\ref{eq:ineq_multi_hopf}) is sufficient for
multistationarity in network~(\ref{eq:MA-seq}).
To be more precise, by \cite[Theorem~5.1]{maya-bistab}, if the
catalytic constants satisfy inequality~(\ref{eq:ineq_multi_hopf}),
then there exists values of the total concentrations of kinase,
phosphatase and protein such that network~(\ref{eq:MA-seq}) has
three positive steady states -- no matter what values the other rate
constants take. As multistationarity is necessary for bistability,
we say in \cite{maya-bistab}, that the catalytic constants enable
the emergence of bistability in sequential and distributive double
phosphorylation.

\subsection{Catalytic constants and non-trivial dynamics}
\label{sec:cat-const-non-triv}

In the previous subsections we have described how the catalytic
constants of cyclic distributive double phosphorylation enable the
emergence of oscillations, and how the catalytic constants of
sequential distributive double phosphorylation enable the emergence
of bistability. Hence we conclude that in distributive double
phosphorylation the catalytic constants enable non-trivial
dynamics.

As a consequence, if the rate constant are chosen according to the
procedure of Section~\ref{sec:procedure} and
Theorem~\ref{theo:main-cyc-full} and network~(\ref{eq:MA_cyc})
admits a stable limit cycle for these rate constants, then
network~(\ref{eq:MA-seq})  taken with the {\em same rate 
  constant values} will show multistationarity -- for some, usually
different, value of the total concentrations. That is, if the
catalytic constants satisfy~(\ref{eq:ineq_multi_hopf}), then a {\em
  cyclic mechanism}  can show {\em sustained oscillations}, while a
{\em sequential mechanism} equipped with the same rate constant
values can show {\em bistability}. 

As an example the procedure described in Section~\ref{sec:procedure}
together with Theorem~\ref{theo:main-cyc-full} have been used to
obtain the following rate constant values:
\begin{align}
  \notag
  \kappa_{1} &= 49 &
    \kappa_{2} &= \frac{1}{10} &
    \kappa_{3} &= \frac{1}{2} &
    \kappa_{4} &= 7 \\
    \label{eq:exa_k_seq_bistab_cyc_osci}
    \kappa_{5} &= \frac{1}{10} &
    \kappa_{6} &=2 &
    \kappa_{7} &=49 &
    \kappa_{8} &=\frac{1}{10} \\
    \notag
    \kappa_{9} &=\frac{1}{4} &
    \kappa_{10} &=7 &
    \kappa_{11}&=\frac{1}{10} &
    \kappa_{12}&=\frac{3}{4}
\end{align}

These values satisfy inequality~(\ref{eq:ineq_multi_hopf}). Using
these values in the ODEs derived from network~(\ref{eq:MA_cyc}), one
detects simple Hopf bifurcations and oscillations as depicted in
Fig.~\ref{fig:cyc_S11_Ktot} and Fig.~\ref{fig:cyc_S00_S11_t}. And
using these values in the ODEs derived from network~(\ref{eq:MA-seq}), one
obtains multistationarity as depicted in Fig.~\ref{fig:seq_S2_Ktot}.
To create these figures, the same parameter values have been
used in both ODE systems, albeit for different values of the total
concentrations. For Fig.~\ref{fig:cyc_S11_Ktot} and
Fig.~\ref{fig:cyc_S00_S11_t} the procedure of
Section~\ref{sec:procedure} yields 
\begin{equation}
  \label{fig:tab_ci_cyc}
  c_1 = \frac{37}{14},\;  c_2 =\frac{115}{21} \text{ and }
  c_3=\frac{425}{42},
\end{equation}
where $c_1$ denotes total amount of kinase~$K$, $c_2$ of
phosphatase~$F$ and $c_3$ of substrate~$S$.
And for Fig.~\ref{fig:seq_S2_Ktot} following the results of
\cite{maya-bistab} yields (using the same notation for the total
concentrations) 
\begin{equation}
  \label{fig:tab_ci_seq}
  c_1=\frac{307}{27},\;  c_2=\frac{650}{27} \text{ and }
  c_3=\frac{18539}{540}\ .
\end{equation}

\begin{figure}[!h]
  \centering

  \begin{subfigure}{0.3\textwidth}
    \centering
    \includegraphics[width=0.95\textwidth]{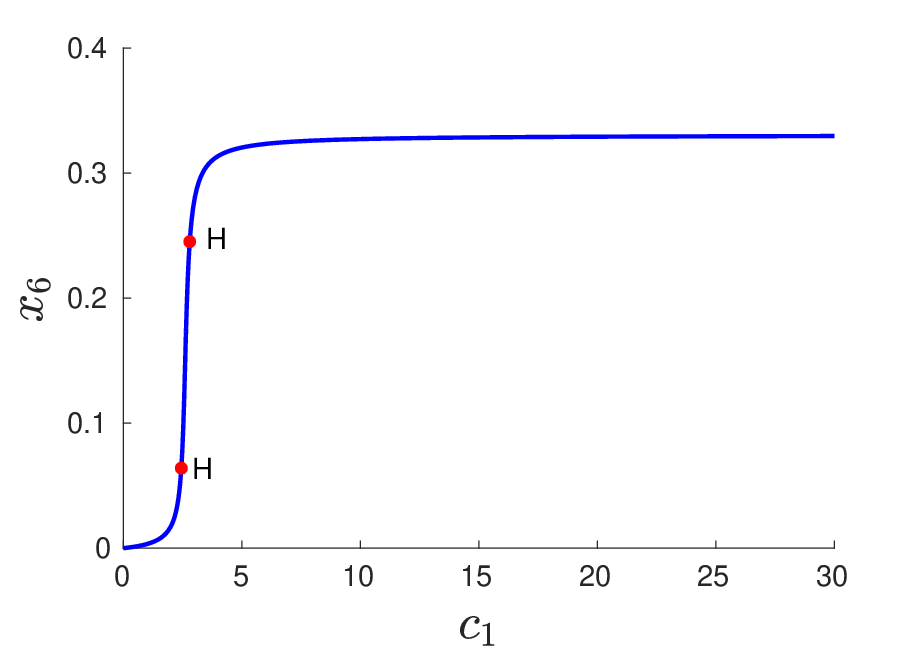}
    \subcaption{
      \label{fig:cyc_S11_Ktot}
      Cyclic: numerical continuation of steady states. Steady state
      value of $S_{11}$ ($x_6$) vs.\ total concentration of Kinase $K$
      ($c_1$); Hopf points (H).
    }
  \end{subfigure}
  \begin{subfigure}{0.3\textwidth}
    \centering
    \includegraphics[width=0.95\textwidth]{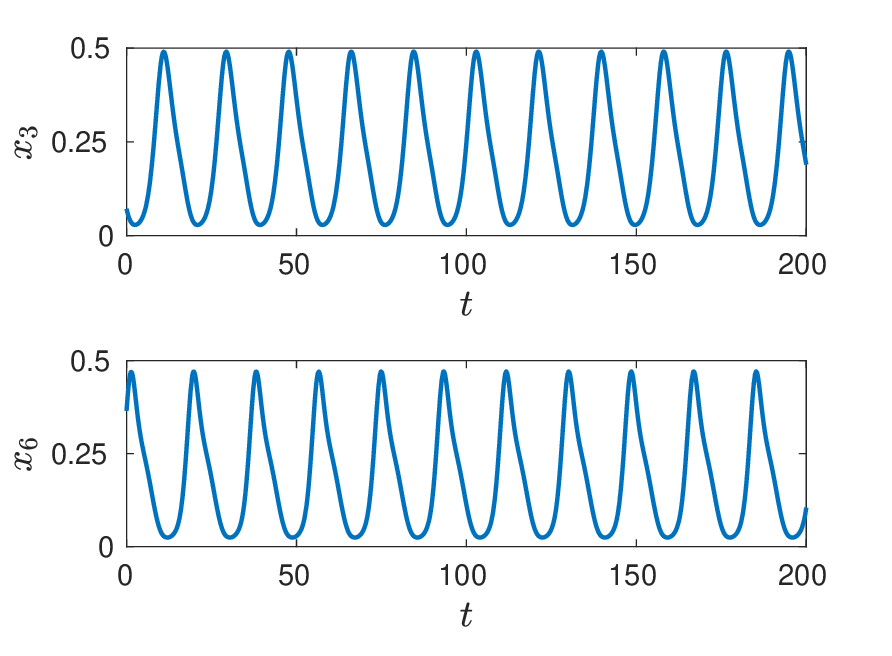}
    \subcaption{
      \label{fig:cyc_S00_S11_t}
      Cyclic: numerical solutions of ODEs defined by
      network~(\ref{eq:MA_cyc}); plotting $S_{00}$ ($x_3$) and
      $S_{11}$ ($x_6$) vs.\ time $t$ shows sustained oscillations.
    }
  \end{subfigure}
  \begin{subfigure}{0.3\textwidth}
    \centering
    \includegraphics[width=0.95\textwidth]{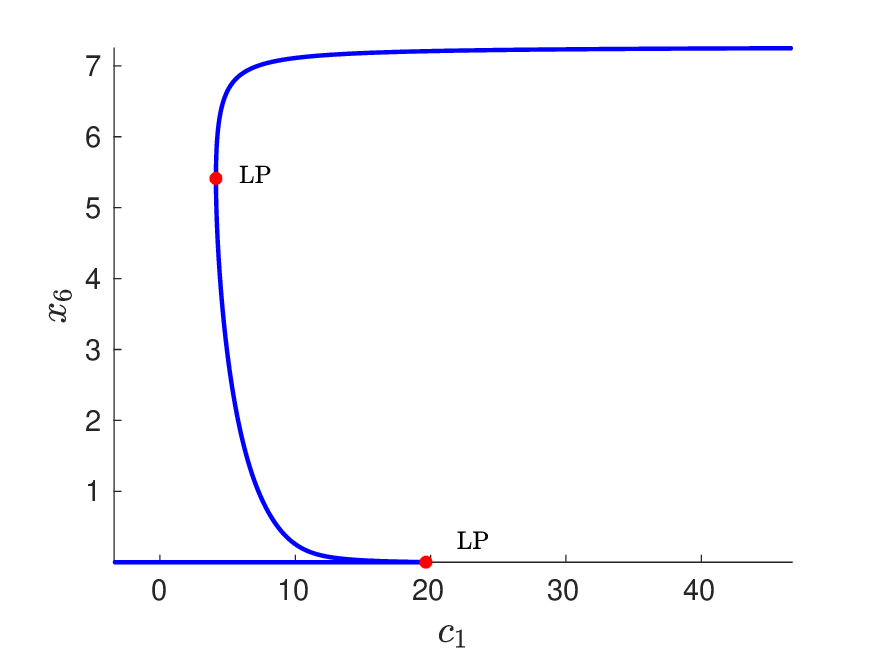}
    \subcaption{
      \label{fig:seq_S2_Ktot}
      Sequential: numerical continuation of steady states; 
      $S_{2}$ ($x_6$) vs.\ $c_1$; LP limit point, multistationarity
      for $c_1$ between LPs.
    }
  \end{subfigure}
  \caption{Oscillations and multistationarity in distributive
    phosphorylation. Panel~(a)~\&~(b) Hopf bifurcations (H) and
    sustained oscillations in network~(\ref{eq:MA_cyc}); panel~(c)
    multistationarity in network~(\ref{eq:MA-seq}). Rate constants
    for both networks as in (\ref{eq:exa_k_seq_bistab_cyc_osci}), total
    concentrations for network~(\ref{eq:MA_cyc})
    in eq.~(\ref{fig:tab_ci_cyc}) and for network~(\ref{eq:MA-seq})
    in eq.~(\ref{fig:tab_ci_seq}).
  }
  \label{fig:Intro_Exa_Cyc_Numerics}
\end{figure}

\section{Data availability}

Data sharing not applicable to this article as no datasets were generated or analyzed during the current study.


\appendix

\section{A Remark on matrices of the form $A=\lambda B$}
\label{App:Rem_A_lamB}

Throughout this section let $A, B\in\R^{n\times n}$ and
$\lambda$ be a nonzero real number such that
\begin{displaymath}
  A = \lambda B
\end{displaymath}
\begin{lemma}
  \label{lem:1}
  For matrices $A$, $B$  as above one has
  \begin{displaymath}
    \det(A) = \lambda^n \det(B)\ .
  \end{displaymath}
\end{lemma}
\begin{proof}
  If $B$ is singular, then $A$ is singular by construction and the
  result follows immediately. Otherwise the result follows by Laplace 
  expansion of $\lambda B$. 
\end{proof}
As in, for example, \cite{Horn2012}, for $A\in\R^{n\times n}$
let $E_k(A)$ denote the sum of the principal minors of size $k$
of the matrix $A$.
\begin{lemma}\label{lem:2}
  For matrices $A$, $B$ as above one has
  \begin{displaymath}
    E_k(A) = \lambda^k E_k(B)\ .
  \end{displaymath}
\end{lemma}
\begin{proof}
  This follows from Lemma~\ref{lem:1} and the fact that $E_k(A)$ is a
  sum of determinants of $k\times k$-sub-matrices.
\end{proof}
Let $a_k$ denote the coefficients of the characteristic polynomial of
$A$ and $b_k$ those of the characteristic polynomial of~$B$.
\begin{lemma}
  \label{lem:3}
  For $A$, $B$ as above 
  \begin{displaymath}
    a_k = \lambda^k b_k\ .
  \end{displaymath}
\end{lemma}
\begin{proof}
  This follows from Lemma~\ref{lem:2} and
  and the fact that the coefficients $a_k$ of any $n\times n$ matrix
  can be defined in terms of the sums of $k\times k$ principal minors
  of $A$ (cf.~\cite[eq.~(1.2.13)]{Horn2012}):
  \begin{displaymath}
    a_k=(-1)^{n-k} E_k(A)\ .
  \end{displaymath}
\end{proof}
And finally:
\begin{lemma}
  \label{lem:4}
  Let $A$, $B$ be as above and assume that
  $\rank(A)=\rank(B)=s<n$. Then the characteristic polynomial of $A$
  can be expressed in terms of the characteristic polynomial of $B$ by
  the following formula:
  \begin{align}
    \label{eq:formula_cp_A_lamB}
    \det(\mu I_n - \lambda A) &= \mu^{n-s} \left(
                                \mu^s + \lambda b_1 \mu^{s-1} + \ldots + \lambda^{s-1} b_{s-1}
                                \mu + \lambda^ s b_s 
                                \right)\\
    \label{eq:sum_formula_cp_A_lamB}
                              &= \mu^{n-s} \lambda^s
                                \left(
                                \left(\frac{\mu}{\lambda}\right)^s
                                + \sum_{i=1}^s \lambda^i b_i(h)
                                \left(\frac{\mu}{\lambda}\right)^{s-i} 
                                \right)
  \end{align}
\end{lemma}
\begin{proof}
  Eq.~(\ref{eq:formula_cp_A_lamB}) follows from Lemma~\ref{lem:2} and
  \cite[eq.~(1.2.13)]{Horn2012}, eq.~(\ref{eq:sum_formula_cp_A_lamB})
  by factoring $\lambda^s$ (which is well defined as in our setting
  $\lambda\neq0$). 
\end{proof}

\section{Hurwitz determinants of $H(\lambda,h)$ and $G(h)$}
\label{sec:hur-hom}

Recall that $s =\rank (S)$, where $S$ is the $n \times m$
stoichiometric matrix. By  Corollary~\ref{coro:cp_J_lamH} we have
$a_i =\lambda^i b_i(h)$, $i=1,2, \ldots , s$.

We use the following formula for the determinant of a matrix $A$
with elements $a_{ij}$, $i$, $j=1$, \ldots $n$
proved in~\cite{maybee1989} 
\begin{equation}\label{eq-cyc-formula}
  \det A = \sum_{\sigma} (-1)^{n+f} A[c_1] \ldots A[c_f],
\end{equation}
where $c_1, \ldots, c_f$ are pairwise disjoint cycles of a permutation
$\sigma$. For a cycle  $c=(i_1, i_2, \ldots ,i_k)$ we have,
\begin{displaymath}
  A[c] = a_{i_2i_1} a_{i_3i_1} \ldots a_{i_1,i_k}\ .
\end{displaymath}
We define the differences between two consecutive indices in a cycle
$c=(i_1, \ldots , i_k)$:
\begin{definition}\label{def:vd}
  Let $c=(i_1 , i_2, \ldots, i_k)$ be a cycle. 
  We define the {\it (cycle)  differences } $d_{i_s i_{s+1}} =i_s
  -i_{s+1}$, $s=1, \ldots ,k$, where $i_{k+1}=i_1$ between any two
  consecutive indices of $c$.  
\end{definition}

The lemma below follows immediately by (\ref{eq-cyc-formula}) and the
differences' definition. In  (\ref{eq:det-c}) we state a different
formulation of the product $H_l[c]$, with the help of  the differences
of a cycle $c$, which is specific to Hurwitz matrices. The same lemma
applies to the Hurwitz matrices $G_l (h)$, $l=1,2, \ldots , s$. 
\begin{lemma}\label{lem:hc} 
  Let   $\det H_{l}$ be the Hurwitz determinant of $l$-th order,
  $l=1,2, \ldots , s$. Then  
  \begin{equation}\label{eq:det}
    \det H_{l} = \sum_{\sigma} (-1)^{l+f} H_{l} [c_1]\ldots H_{l} [c_f]
  \end{equation}
  where $\sigma$ is a permutation and $c_1, \ldots , c_f$ are the  set
  of corresponding pairwise disjoint cycles of $\sigma$. We have  for
  a cycle $c =(i_1,\ldots , i_k)$  and the corresponding product
  \begin{equation}\label{eq:det-c}
    H_{l}[c] =a_{2i_1-i_2} a_{2i_2-i_3} \ldots a_{2i_k-i_1}
    =a_{i_1+d_{i_1 i_2}}  a_{i_2+d_{i_2 i_3}}  \ldots a_{i_k+d_{i_k
        i_1}}\ . 
  \end{equation}
\end{lemma}
The next lemma follows immediately by Definition~\ref{def:vd}.
\begin{lemma}\label{ldp}
Let $c=(i_1, i_2, \ldots ,  i_k)$ be a cycle.
For any cycle $c$, the sum of its differences is zero, 
$d_{i_1 i_{2}} +d_{i_2 i_{3}} +\ldots  +d_{i_k i_{1}} =0\, . $
\end{lemma}

\begin{remark}\label{rem:h-g-l}
  We notice that
  $\det H_l (1,h) =\det G_l (h)$ for $l=1,2, \ldots s$  if $\lambda
  =1$.
\end{remark}

Now we can turn to the proof of Proposition~\ref{pro:hom}:
\begin{proof}[Proof of Proposition~\ref{pro:hom}]
  Let $c=(i_1, \ldots  ,i_k)$ be a cycle. By by
  Corollary~\ref{coro:cp_J_lamH} we have for the product
  $a_{i_1} \ldots a_{i_k}$
  \begin{displaymath}
    a_{i_1} \ldots a_{i_k} = \lambda^{\delta} b_{i_1} \ldots
    b_{i_k}, \text{ where $\delta=i_1 +\ldots + i_k$.}
  \end{displaymath}
  By Lemma~\ref{lem:hc} and in particular~(\ref{eq:det-c}) 
  \begin{equation}
    \begin{split}
      \label{eq:hlc}
      H_l [c]  &= a_{i_1 +d_{i_1 i_2}} \ldots a_{i_k +d_{i_k i_1}}  =
      \lambda^{i_1 +d_{i_1 i_2}} b_{i_1+d_{i_1 i_2}} \ldots
      \lambda^{i_k +d_{i_k i_1}} b_{i_k+d_{i_k i_1}}  \\
      &=\lambda^{\delta} b_{i_1} \ldots b_{i_k}
      = \lambda^{\delta} G_l [c]
    \end{split}
  \end{equation}
  where we have used that the sum of the differences of a cycle sum up
  to zero  by  Lemma~\ref{ldp}.
 
  Let the sum of the indices of each cycle $c_1,, \ldots c_f$ of a
  permutation $\sigma$ be denoted by
  $\delta_1$, \ldots ,$\delta_f$,
  correspondingly.    We  use the fact that   the cycles $\{c_1,
  \ldots ,c_f \}$  in a permutation $\sigma$  are pairwise disjoint.
  Thus it follows that
  \begin{equation}\label{eq:ss}
    \delta_1  + \delta_2  + \ldots + \delta_f = 1 +2 + \ldots + l =
    \frac{l(l+1)}{2}.
  \end{equation} 
  
  By Lemma~\ref{lem:hc}, (\ref{eq:hlc}) and (\ref{eq:ss}) we have
  \begin{align*}
    \det H_l (\lambda,  h) = & \Sigma_{\sigma} (-1)^{l+f} H_l [c_1]
                               \ldots H_l [c_f] \\
                             &= \Sigma_{\sigma} (-1)^{l+f}
                               \lambda^{\delta_1} G_l [c_1] \ldots
                               \lambda^{\delta_f} G_l [c_f]  \\
                             &= \Sigma_{\sigma} (-1)^{l+f}
                               \lambda^{\delta_1 +\ldots +\delta_f}
                               G_l [c_1] \ldots  G_l [c_f]  \\
                             &= \lambda^{l(l+1)/2} \Sigma_{\sigma}
                               (-1)^{l+f} G_l [c_1] \ldots  G_l [c_f]
    \\
                             &= \lambda^{l(l+1)/2} \det G_l (h)\ . 
  \end{align*}
  Thus, $\det H_l (\lambda, h) =\lambda^{l (l+1)/2} \det G_l (h)$ for
  $l=1,2, \ldots , s$.  
\end{proof}

\section{Data for network~(\ref{net:cyc_simple})}
\label{app:jac-h-lambda}

\begin{table}[!h]
  \centering
  \begin{tabular}{|>{$}c<{$}|>{$}c<{$}|} \hline
    S_{00} + K  & y^{(1)} =(1, 0, 1, 0, 0, 0, 0, 0, 0, 0)^T \\ \hline
    S_{00} K & y^{(2)} =( 0, 0, 0, 0, 0, 0, 1, 0, 0, 0)^T \\ \hline
    S_{10} + K & y^{(3)} =( 1, 0, 0, 1, 0, 0, 0, 0, 0, 0)^T \\ \hline
    S_{10} K & y^{4)} =(  0, 0, 0, 0, 0, 0, 0, 1, 0, 0)^T \\ \hline
    S_{11} + K & y^{(5)} =(1,0,0,0,0,1,0,0,0,0)^T \\ \hline
    S_{11} + F & y^{(6)} =( 0, 1, 0, 0, 0, 1, 0, 0, 0, 0)^T \\ \hline
    S_{11} F & y^{(7)} =( 0, 0, 0, 0, 0, 0, 0, 0, 0, 1)^T \\ \hline
    S_{01} + F & y^{(8)} =( 0, 1, 0, 0, 1, 0, 0, 0, 0, 0)^T \\ \hline
    S_{01} F & y^{(9)} =( 0, 0, 0, 0, 0, 0, 0, 0, 1, 0)^T \\ \hline
    S_{00} + F & y^{(10)} =( 0, 1, 1, 0, 0, 0, 0, 0, 0, 0)^T \\ \hline
  \end{tabular}
  \caption{Complexes of network~(\ref{eq:MA_cyc})
    and~(\ref{net:cyc_simple}).}
  \label{tab:complexes-cyc-and-cyc-simple}
\end{table}
The stoichiometric matrix~$S$, the exponent matrix~$Y$ and a
matrix~$W$ defining the conservation relations:
\begin{displaymath}
  \begin{split}
    \tiny
    S &= \left[
      \begin{array}{rrrrrrrr}
        -1 & 1 & -1 & 1 & 0 & 0 & 0 & 0 \\
        0 & 0 & 0 & 0 & -1 & 1 & -1 & 1 \\
        -1 & 0 & 0 & 0 & 0 & 0 & 0 & 1 \\
        0 & 1 & -1 & 0 & 0 & 0 & 0 & 0 \\
        0 & 0 & 0 & 0 & 0 & 1 & -1 & 0 \\
        0 & 0 & 0 & 1 & -1 & 0 & 0 & 0 \\
        1 & -1 & 0 & 0 & 0 & 0 & 0 & 0 \\
        0 & 0 & 1 & -1 & 0 & 0 & 0 & 0 \\
        0 & 0 & 0 & 0 & 0 & 0 & 1 & -1 \\
        0 & 0 & 0 & 0 & 1 & -1 & 0 & 0 \\
      \end{array}
    \right],\; 
    Y= \left[
      \begin{array}{rrrrrrrr}
        1 & 0 & 1 & 0 & 0 & 0 & 0 & 0 \\
        0 & 0 & 0 & 0 & 1 & 0 & 1 & 0 \\
        1 & 0 & 0 & 0 & 0 & 0 & 0 & 0 \\
        0 & 0 & 1 & 0 & 0 & 0 & 0 & 0 \\
        0 & 0 & 0 & 0 & 0 & 0 & 1 & 0 \\
        0 & 0 & 0 & 0 & 1 & 0 & 0 & 0 \\
        0 & 1 & 0 & 0 & 0 & 0 & 0 & 0 \\
        0 & 0 & 0 & 1 & 0 & 0 & 0 & 0 \\
        0 & 0 & 0 & 0 & 0 & 0 & 0 & 1 \\
        0 & 0 & 0 & 0 & 0 & 1 & 0 & 0 \\
      \end{array}
    \right],\\
    W^T &=
    \begin{bmatrix}
      1&0&0&0&0&0&1&1&0&0\\
      0&1&0&0&0&0&0&0&1&1\\
      0&0&1&1&1&1&1&1&1&1\\
    \end{bmatrix}
  \end{split}
\end{displaymath}
The vector of rate functions is
\begin{equation}
  \label{eq:v_cyc_E}
  v(k,x) = (\kappa_1 x_1 x_3, \kappa_3x_7, \kappa_4x_1x_4, \kappa_6x_8, \kappa_7 x_2 x_6,\kappa_9x_{10},\kappa_{10}x_2x_5,\kappa_{12}x_9)^T \; .
\end{equation}
The diagonal matrix of rate constants
\begin{displaymath}
\diag(k)= \left[
    \begin{array}{cccccccc}
      k_{1} & 0 & 0 & 0 & 0 & 0 & 0 & 0  \\
      0 & k_{2} & 0 & 0 & 0 & 0 & 0 & 0  \\
      0 & 0 & k_{3} & 0 & 0 & 0 & 0 & 0  \\
      0 & 0 & 0 & k_{4} & 0 & 0 & 0 & 0 \\
      0 & 0 & 0 & 0 & k_{5} & 0 & 0 & 0 \\
      0 & 0 & 0 & 0 & 0 & k_{6} & 0 & 0\\
      0 & 0 & 0 & 0 & 0 & 0 & k_{7} & 0\\
      0 & 0 & 0 & 0 & 0 & 0 & 0 & k_{8}  \\
      \end{array}
  \right]
\end{displaymath}
The vector of monomials of the rate function $v(k,x)$
\[
\phi(k,x) = (x_1 x_3, x_7, x_1x_4, x_8,  x_2 x_6,x_{10},x_2x_5,x_9)^T \; .
\]

\section{Initial data for Fig.~\ref{fig:cyc_S11_Ktot} \& 
  \ref{fig:cyc_S00_S11_t} }
\label{sec:initial_data}

To generate Fig.~\ref{fig:cyc_S11_Ktot} \& \ref{fig:cyc_S00_S11_t} the
ODEs derived from the full (reversible) reaction
network~(\ref{eq:MA_cyc}) are used (i.e.\ the
ODEs~(\ref{eq:sys_cyc_full_1}) -- (\ref{eq:sys_cyc_full_10})).  
In both Fig.~\ref{fig:cyc_S11_Ktot} and \ref{fig:cyc_S00_S11_t} we use
the point displayed in Table~\ref{tab:ini_fig_Intro} as initial value. 
\begin{table}[!h]
  \centering
  \begin{tabular}{|c|c|c|c|c|c|c|c|c|c|}\hline
    $x_1$ & $x_2$ & $x_3$ & $x_4$ & $x_5$ & $x_6$ & $x_7$ & $x_8$
    & $x_9$ & $x_{10}$ \\ \hline 
    $\frac{1}{7}$ & $\frac{1}{7}$ & $\frac{1}{7}$ & $1$ & $1$
    & $\frac{1}{7}$ & $2$ &$\frac{1}{2}$ & $\frac{4}{3}$ & $4$ \\ \hline
  \end{tabular}
  \caption{Initial value for Fig.~\ref{fig:cyc_S11_Ktot} \&
    \ref{fig:cyc_S00_S11_t}} 
  \label{tab:ini_fig_Intro}
\end{table}

\begin{remark}
  The point in Table~\ref{tab:ini_fig_Intro} is a steady state of the
  irreversible network~(\ref{net:cyc_simple}) generated according to
  our procedure (with $h_7=\frac{1}{2}$, $h_8=2$, $h_9=\frac{3}{4}$,
  $h_{10}=\frac{1}{4}$, $h_4=h_5=1$ and $t=1$).

  We use Matlab's \texttt{ode15s} to solve the initial value problem
  defined by the above ODEs with initial value given in
  Table~\ref{tab:ini_fig_Intro} and  
  the backward constants 
  $\kappa_2=\kappa_5=\kappa_8=\kappa_{11}=\frac{1}{10}$ (and the
  remaining constants according to our procedure for the
  irreversible network).
  These backward constants are \lq small enough\rq{} in the sense of
  the results of \cite{Banaji2018},  that is, that the reversible
  network~(\ref{eq:MA_cyc}) has a limit cycle close to the limit cycle
  of the irreversible network.

  However, on the one hand the steady state of
  Table~\ref{tab:ini_fig_Intro} 
  (of the irreversible network) is \lq far 
  enough\rq{} from a steady state of the reversible network in the
  following sense: the solution of the reversible system with initial
  value given in Table~\ref{tab:ini_fig_Intro} approaches a limit
  cycle of the reversible system (if approximated with
  \texttt{ode15s}). If the point given in 
  Table~\ref{tab:ini_fig_Intro} was \lq too close\rq{} to a steady
  state of the reversible system the solution with it as initial value
  would approach that steady state (if approximated with
  \texttt{ode15s}).

  On the other hand, Fig.~\ref{fig:cyc_S11_Ktot} is generated with
  \texttt{Matcont} 
  using the point given in Table~\ref{tab:ini_fig_Intro} as an initial
  guess of a steady state of the above ODEs (of the reversible
  network). And this point is close enough to a steady state of the
  reversible network for \texttt{Matcont} to converge to a steady
  state. 
\end{remark}

\section{Initial data for Fig.~\ref{fig:seq_S2_Ktot}}
\label{sec:Initial_data_Intro_Exa_Seq}

To obtain Fig.~\ref{fig:seq_S2_Ktot} we
follow~\cite{maya-bistab}, where a system of nine ODEs is derived from
network~(\ref{eq:MA-seq}). As in this
reference, we use the variables given in
Table~\ref{tab:vars_for_Intro_Exa_Seq} to denote the species
concentrations. As this network does not distinguish $S_{10}$ and
$S_{01}$, there is only one mono-phosphorylated from of the substrate
($S_1$) and $S_2$ is used to denote the double-phosphorylated
substrate. Consequently, this network contains only nine species. It
has, however, $12$ reactions and the labeling is consistent with
network~(\ref{eq:MA_cyc}).

\begin{table}[!h]
  \centering
  \begin{tabular}{|c|c|c|c|c|c|c|c|c|}\hline
    $x_1$ & $x_2$ & $x_3$ & $x_4$ & $x_5$ & $x_6$ & $x_7$ & $x_8$
    & $x_9$  \\ \hline 
    $S_0$ & $K$ & $K S_0 $ & $S_1$ & $K S_1$ & $S_2$ & $F$
    & $F S_2$ & $F S_1 $ \\ \hline
  \end{tabular}
  \caption{Variables and species for network~(\ref{eq:MA-seq}) (as
    in~\cite{maya-bistab}).}
  \label{tab:vars_for_Intro_Exa_Seq}
\end{table}
Using the approach described in \cite{maya-bistab} we obtain the steady
state depicted in Table~\ref{tab:ini_for_Intro_Exa_Seq}.  We use this point as
a starting point for the numerical continuation in \texttt{Matcont}. 

\begin{table}[!h]
  \centering
  \begin{tabular}{|c|c|c|c|c|c|c|c|c|}\hline
    $x_1$ & $x_2$ & $x_3$ & $x_4$ & $x_5$ & $x_6$ & $x_7$ & $x_8$
    & $x_9$  \\ \hline 
    $\frac{1}{10}$ & $1$ & $\frac{49}{6}$ & $\frac{119}{180}$
    & $\frac{119}{54}$ & $\frac{17}{135}$ & $1$ & $\frac{476}{27}$
    & $\frac{49}{9}$ \\ \hline
  \end{tabular}
  \caption{
    Starting point for the numerical continuation in
    Fig.~\ref{fig:seq_S2_Ktot} (a steady state of network~(\ref{eq:MA-seq})
    for the rate constant values of
    eq.~(\ref{eq:exa_k_seq_bistab_cyc_osci}) generated as described
    in~\cite{maya-bistab}).
  }
  \label{tab:ini_for_Intro_Exa_Seq}
\end{table}

\subsection*{Acknowledgments}
Maya Mincheva wishes to thank the Institute of Mathematics and
Informatics at the Bulgarian Academy of Sciences where part of this
research was done  for their hospitality.

Carsten Conradi was partially funded by DFG Grant 517274113.

MM and CC thank all reviewers for their diligent review and helpful
suggestions.



\providecommand{\bysame}{\leavevmode\hbox to3em{\hrulefill}\thinspace}
\providecommand{\MR}{\relax\ifhmode\unskip\space\fi MR }
\providecommand{\MRhref}[2]{%
  \href{http://www.ams.org/mathscinet-getitem?mr=#1}{#2}
}
\providecommand{\href}[2]{#2}

\end{document}